\documentclass[11pt]{article}
\usepackage{empheq}
\usepackage{xcolor}

\usepackage{graphics} % for pdf, bitmapped graphics files
\usepackage{epsfig} % for postscript graphics files
\let\proof\relax
\let\endproof\relax
\usepackage{amsthm}  % assumes amsmath package iMmnstalled
\usepackage{amscd}
\usepackage[noadjust]{cite}
\usepackage{float}
\usepackage{times}
\usepackage{color}
\usepackage[utf8]{inputenc}
\usepackage[english]{babel}
\usepackage{algorithm}
\usepackage{multicol}

\usepackage{algpseudocode}
\newtheorem{theorem}{Theorem}
\newtheorem{lemma}{Lemma}

\newtheorem{proposition}{Proposition}

\usepackage{epsfig}
\usepackage{amsfonts}
\usepackage{mathrsfs}
\usepackage{amssymb}
\usepackage{amstext}
\usepackage{amsmath}
\usepackage{xspace}
\usepackage{comment}
\usepackage{multirow}
\usepackage{setspace}
\usepackage[compact]{titlesec}
\usepackage{comment}
\usepackage{version}
\usepackage{enumitem}
\usepackage{multirow,subfigure}
\usepackage[pagebackref=true, bookmarks]{hyperref}
\hypersetup{
pageanchor=true,
bookmarksnumbered,
colorlinks=true,
menubordercolor= [rgb]{0, 0, 1},
linkbordercolor= [rgb]{0, 1, 1},
citecolor = red,
linkcolor= blue}

\usepackage{url}
\usepackage[sort]{natbib}

\usepackage{my_definitions}

\makeatletter
\renewcommand{\section}{\@startsection{section}{1}{0mm}%
                                   {2ex plus -1ex minus -.2ex}%
                                   {1.3ex plus .2ex}%
                                   {\normalfont\Large\bfseries}}%
 \renewcommand{\subsection}{\@startsection{subsection}{2}{0mm}%
                                     {1ex plus -1ex minus -.2ex}%
                                     {1ex plus .2ex }%
                                     {\normalfont\large\bfseries}}%
 \renewcommand{\subsubsection}{\@startsection{subsubsection}{3}{0mm}%
                                     {1ex plus -1ex minus -.2ex}%
                                     {1ex plus .2ex }%
                                     {\normalfont\normalsize\bfseries}}
 \renewcommand\paragraph{\@startsection{paragraph}{4}{0mm}%
                                    {1ex \@plus1ex \@minus.2ex}%
                                    {-1em}%
                                    {\normalfont\normalsize\bfseries}}
\titlespacing{\section}{2pt}{*0}{*0}
\titlespacing{\subsection}{1pt}{*0}{*0}
\titlespacing{\subsubsection}{0pt}{*0}{*0}

\makeatother

\allowdisplaybreaks

\graphicspath{{./figures/}}

\includeversion{paperonly}
\excludeversion{techreportonly}
\begin{document}
	\title{Approximation Schemes for Multiperiod Binary Knapsack Problems} 
	\author{Zuguang Gao, John R. Birge, and Varun Gupta\footnote{All authors are with the University of Chicago. Emails: \{zuguang.gao, john.birge, varun.gupta\}@chicagobooth.edu.}}
	\date{}

	\maketitle
	
	\begin{abstract}
		\begin{onehalfspace} 
			An instance of the multiperiod binary knapsack problem (MPBKP) is given by a horizon length $T$, a non-decreasing vector of knapsack sizes $(c_1, \ldots, c_T)$ where $c_t$ denotes the cumulative size for periods $1,\ldots,t$, and a list of $n$ items. Each item is a triple $(r, q, d)$ where $r$ denotes the reward or value of the item, $q$ its size, and $d$ denotes its time index (or, deadline). The goal is to choose, for each deadline $t$, which items to include to maximize the total reward, subject to the constraints that for all $t=1,\ldots,T$, the total size of selected items with deadlines at most $t$ does not exceed the cumulative capacity of the knapsack up to time $t$. We also consider the multiperiod binary knapsack problem with soft capacity constraints (MPBKP-S) where the capacity constraints are allowed to be violated by paying a penalty that is linear in the violation. The goal of MPBKP-S is to maximize the total profit, which is the total reward of the selected items less the total penalty. Finally, we consider the multiperiod binary knapsack problem with soft stochastic capacity constraints (MPBKP-SS), where the non-decreasing vector of knapsack sizes $(c_1, \ldots, c_T)$ follow some arbitrary joint distribution but we are given access to the profit as an oracle, and we must choose a subset of items to maximize the total expected profit, which is the total reward less the total expected penalty.

%			In the multiperiod binary knapsack problem (MPBKP), there are $T$ time periods. In each period, there are a number of items, each with a reward and a size. The goal is to choose in each period which items to include to maximize the total reward, subject to the constraints that for any $t=1,\ldots,T$, the total size of selected items from period~$1$ to period~$t$ cannot exceed the capacity of the knapsack at time $t$. In the multiperiod binary knapsack problem with soft capacity constriants (MPBKP-SS), the capacity constraints are allowed to be violated by paying some penalty for each unit size that goes beyond the capacity. The goal of MPBKP-SS is then to maximize the total profit, which is the total reward of the selected items deducted by the total penalty.

For MPBKP, we exhibit a fully polynomial-time approximation scheme that achieves $(1+\epsilon)$ approximation with runtime $\tilde{\mathcal{O}}\left(\min\left\{n+\frac{T^{3.25}}{\epsilon^{2.25}},n+\frac{T^{2}}{\epsilon^{3}},\frac{nT}{\epsilon^2},\frac{n^2}{\epsilon}\right\}\right)$; for MPBKP-S, the $(1+\epsilon)$ approximation can be achieved in $\mathcal{O}\left(\frac{n\log n}{\epsilon}\cdot\min\left\{\frac{T}{\epsilon},n\right\}\right)$. To the best of our knowledge, our algorithms are the first FPTAS for any multiperiod version of the Knapsack problem since its study began in 1980s. For MPBKP-SS, we prove that a natural greedy algorithm is a $2$-approximation when items have the same size. Our algorithms also provide insights on how other multiperiod versions of the knapsack problem may be approximated.

%\keywords{approximation algorithms \and knapsack problem \and optimization.}
		\end{onehalfspace}
	\end{abstract}
	
	%\tableofcontents
	\thispagestyle{empty}
	\setlength{\parskip}{.1in}
	\maketitle

	\clearpage
	\setcounter{page}{1}

	\section{Introduction}
	
	\label{sec:typesetting-summary}
Knapsack problems are a classical category of combinatorial optimization problems, and have been studied for more than a century~\citep{mathews1896partition}. They have found wide applications in various fields~\citep{strusevich2005knapsack}, such as selection of investments and portfolios, selection of assets, finding the least wasteful way to cut raw materials, etc.
One of the most commonly studied problem is the so-called \emph{0-1 knapsack problem}, where a set of $n$ items are given, each with a reward and a size, and the goal is to select a subset of these items to maximize the total reward, subject to the constraint that the total size may not exceed some knapsack capacity. It is well-known that the 0-1 knapsack problem is NP-complete. However, the problem was shown to possess \emph{fully polynomial-time approximation schemes (FPTASs)}, i.e., there are algorithms that achieve $(1+\epsilon)$ factor of the optimal value for any $\epsilon\in (0,1)$, and take polynomial time in $n$ and $1/\epsilon$.

The first published FPTAS for the 0-1 knapsack problem was due to~\cite{ibarra1975fast} where they achieve a time complexity $\widetilde{\Ocal}\left(n+(1/\epsilon^4)\right)$ by dividing the items into a class of ``large'' items and a class of ``small'' items. The problem is first solved for large items only, using the dynamic program approach, with rewards rounded down using some discretization quantum (chosen in advance), and the small items are added later. \cite{lawler1979fast} proposed a more nuanced discretization method to improve the polylogarithmic factors. Since then, improvements have been made on the dynamic program for large items~\citep{kellerer2004improved,rhee2015faster}. Most recently, the FPTAS has been improved to $\widetilde{\Ocal}\left(n + (1/\epsilon)^{9/4}\right)$ in~\cite{jin:LIPIcs:2019:10652}.

In this paper, we study {  three} extensions of the 0-1 knapsack problem. First, we consider a multiperiod version of the 0-1 knapsack problem, which we call the \emph{multiperiod binary knapsack problem (MPBKP)}. 
There is a horizon length $T$ and a vector of knapsack sizes $(c_1,\ldots,c_T)$, where $c_t$ is the cumulative size for periods $1,\ldots,t$ and is non-decreasing in $t$. We are also given a list of $n$ items, each associated with a triple $(r, q, d)$ where $r$ denotes the reward or value of the item, $q$ its size, and $d$ denotes its time index (or, deadline).
The goal is to choose a reward maximizing set of items to include such that for any $t=1,\ldots,T$, the total size of selected items with deadlines at most $t$ does not exceed the cumulative capacity of the knapsack up to time $t$. The application that motivates this problem is a seller who produces $(c_t - c_{t-1})$ units of a good in time period $t$, and can store unsold goods for selling later. The seller is offered a set of bids, where each bid includes a price ($r$), a quantity demanded ($q$), and a time at which this quantity is needed. The problem of deciding the revenue maximizing subset of bids to accept is exactly MPBKP.

The second extension we consider is the \emph{multiperiod binary knapsack problem with soft capacity constraints (MPBKP-S)} where {  at each period the capacity constraint is allowed to be violated by paying a penalty that is linear in the violation.} The goal of MPBKP-S is then to maximize the total profit, which is the total reward of the selected items less the total penalty. In this case, the seller can procure goods from outside at a certain rate if his supply is not enough to fulfill the bids he accepts, and wants to maximize his profit.

{ 
The third extension we consider is the \emph{multiperiod binary knapsack problem with soft stochastic capacity constraints (MPBKP-SS)} where the non-decreasing vector of knapsack sizes $(c_1, \ldots, c_T)$ follows some arbitrary joint distribution given as the set of sample paths of the possible realizations and their probabilities. We select the items \emph{before} realizations of any of these random incremental capacities to maximize the total \emph{expected} profit, which is the total reward of selected items less the total expected penalty. In this case, the production of the seller at each time is random, but he has to select a subset of bids before realizing his supply. Again, the seller can procure capacity from outside at a certain rate if his realized supply is not enough to fulfill the bids he accepts, and wants to maximize his expected profit.
}

We note that MPBKP is also related to a number of other multiperiod versions of the knapsack problem in literature. The multiperiod knapsack problem (MPKP) proposed by~\cite{faaland1981multiperiod} has the same structure as  MPBKP, except that in~\cite{faaland1981multiperiod}, each item can be repeated multiple times, i.e., the decision variables for each item is not binary, but any nonnegative integer (in the single-period case, this is called the unbounded knapsack problem~\citep{andonov2000unbounded}). To the best of our knowledge, there has been no further studies on MPKP since~\cite{faaland1981multiperiod}. In the multiple knapsack problem (MKP), there are $m$ knapsacks, each with a different capacity, and items can be inserted to any knapsacks (subject to its capacity constraints). It has been shown in~\cite{chekuri2005polynomial} that MKP does not admit an FPTAS, but an efficient polynomial time approximation scheme (EPTAS) has been found in~\cite{10.1007/978-3-642-27660-6_26}, with runtime depending polynomially on $n$ but exponentially on $1/\epsilon$. The incremental knapsack problem (IKP) is another multiperiod version of the knapsack problem~\citep{10.1007/11764298_4}, where the knapsack capacity increases over time, and each selected item generates a reward on every period after its insertion, but this reward is discounted over time. Unlike MPBKP, items do not have deadlines and can be selected anytime throughout the $T$ periods. A PTAS for the IKP when the discount factor is~$1$ (time invariant, referred to as IIKP) and $T=\Ocal\left(\sqrt{\log n}\right)$ has been found in~\cite{bienstock2013approximation}, and it has been shown that IIKP is strongly NP-hard. Later,~\cite{faenza2018ptas} proposed the first PTAS for IIKP regardless of $T$, and~\cite{della2019approximating} proposed an PTAS for IKP when $T$ is a constant. Most recent developments of IKP include~\cite{aouad2020approximate,faenza2020approximation}. Other similar problems and/or further extensions include the multiple-choice multiperiod knapsack problem~\citep{randeniya1994multiple,lin2004multiple,lin2010dynamic}, the multiperiod multi-dimensional knapsack problem~\citep{lau2004multi}, the multiperiod precedence-constrained knapsack problem~\citep{moreno2010large,samavati2017methodology}, to name a few.

Our main contributions of this paper are two-fold. First, from the perspective of model formulation, we propose the MPBKP and its generalized versions MPBKP-S and MPBKP-SS. {  Despite the fact that there are a number of multiperiod/multiple versions of knapsack problems, including those mentioned above (many of which are strongly NP-hard), the MPBKP and MPBKP-S we proposed here are the first to admit an FPTAS among any multiperiod versions of the classical knapsack problem since their initiation back in 1980s.} With these results, it is thus interesting to see where the boundary lies between these multiperiod problems that admit an FPTAS and those problems that do not admit an FPTAS. { Second, the algorithms we propose for both MPBKP and MPBKP-S are generalized from the ideas of solving 0-1 knapsack problems, but with nontrivial modifications as we will address in the following sections. For MPBKP-SS, we propose a greedy algorithm that achieves $2$-approximation for the special case when all items have the same size.}

The rest of this paper is organized as follows. In Section~\ref{sec:form} we formally write the three problems in mathematical programming form. The FPTAS for MPBKP is proposed in Section~\ref{sec:MPBKP} and the FPTAS for MPBKP-S is proposed in Section~\ref{sec:approx2}. Alternative algorithms for both problems are also provided in Apendix. A greedy algorithm for a special case of MPBKP-SS is proposed in Section~\ref{sec:unit-MPBKP-SS}. All proofs are left to Appendix but we provide proof ideas in the main body.

%	\subsection{Notation}	
%	\[ \Qcal(\Scal) = \sum_{i \in \Scal} q_i \]
%	\[ \Rcal(\Scal) = \sum_{i \in \Scal} r_i \]
%	\[ \hat{\Rcal}(\Scal) = \sum_{i \in \Scal} \hat{r}_i \]
%	\[ \Pcal(\Scal) = \Rcal(\Scal) - \sum_{t=1}^T B_t \left[\sum_{j\in  \Scal : d_j = t } q_j - \max_{0 \leq t' < t}\left\{ c_t - c_{t'} - \sum_{j\in \Scal : t'+1 \leq d_j \leq  t-1}q_j \right\}\right]^+ \]
	
%	For a solution $\Scal =  \Scal(1) \cup \Scal(2) \cup \cdots \Scal(T)$ with $\Scal(t) = (i^{(t)}_1, \ldots, i^{(t)}_{I_t})$ denoting an indexing of items in the solution $\Scal$ with deadline $t$, we define the leftover capacity for serving items in $\Scal$ via the Lindley-type recursion:
%	\begin{align*}
%	W^{(1)}_1 &= c_1 \\
%	W^{(t)}_1 &= \left( W^{(t-1)}_{I_{t-1}} - q_{i^{(t-1)}_{I_{t-1}}}\right)^+  + (c_t-c_{t-1})   & (t \geq 2)\\
%	W^{(t)}_j &= \left( W^{(t)}_{j-1} - q_{i^{(t)}_{j-1}} \right)^+   & (j \geq 2).
%	\end{align*}
	
%	Based on $W^{(t)}_j$ defined above, we then define the rounded profit of the set $\Scal$ as the sum of the rounded (down) reward of each item in $\Scal$ minus the rounded (up) penalty for each item in $\Scal$ (the discretization quantum $\kappa$ will be clear from the context and hence we suppress the dependence of $\hat{\Pcal}$ on it):
%	
%	\[ \hat{\Pcal}(\Scal) = \hat{
%		\Rcal}(\Scal) - \sum_{t=1}^T \sum_{j =1 }^{I_t} \left\lceil B_t \cdot \left( q^{(t)}_j -  W^{(t)}_j \right)^+ \right \rceil_\kappa . \]
	
	%\clearpage
\section{Problem Formulation and Main Results}\label{sec:form}
In this section, we formally introduce the Multiperiod Binary Knapsack Problem (MPBKP), as well as the generalized versions: the Multiperiod Binary Knapsack Problem with Soft capacity constraints (MPBKP-S), {  and Multiperiod Binary Knapsack Problem with Soft Stochastic Capacity constraints (MPBKP-SS)}.

%	\varun{Maybe we can mention that intuitively this is $T$ knapsack problems with capacity $c_t-c_{t-1}$ for the $t$th problem but where we can (a) carry forward unused capacity, or (b) additionally buy extra capacity at cost $B_t$ (and optionally carry that forward too) }

\subsection{Multiperiod binary Knapsack problem (MPBKP)}

An instance of MPBKP is given by a set of $n$ items, each associated with a triple $(r_i,q_i,d_i)$, and a sequence of knapsack capacities $\{c_1,\ldots,c_T\}$. For each item $i$, we get reward $r_i$ if and only if $i$ is included in the knapsack by time $d_i$. We assume that $r_i\in \NN$, $q_i\in\NN$ and $d_i\in [T]:=\{1,\ldots, T\}$. The knapsack capacity at time $t$ is $c_t$, and by convention $c_0=0$. The MPBKP can be written in the integer program (IP) form:
\begin{subequations}\label{MPBKP}
	\begin{align}
	&\max_x z = \sum_{i=1}^{n} r_ix_i\\
	&\text{ s.t. } \sum_{j: d_j\le t} q_jx_j\le c_{t},\quad \forall t=1,\ldots, T\\
	&\qquad x_i\in\{0,1\},\quad \forall i=1,\ldots,n
	\end{align}
\end{subequations}
where $x_i$'s are binary decision variables, i.e., $x_i$ is~$1$ if item $i$ is included in the knapsack and is~$0$ otherwise. In~\eqref{MPBKP}, we aim to pick a subset of items to maximize the objective function, which is the total reward of picked items, subject to the constraints that by each time $t$, the total size of picked items with deadlines up to $t$ does not exceed the knapsack capacity at time $t$, which is $c_{t}$. 
For each $t\in [T]$, let $\mathcal{I}(t):=\{i\in [n]\mid d_i=t\}$ denote the set of items with deadline $t$. Note that without loss of generality, we may assume that $\mathcal{I}(t)\ne\emptyset, \forall t$ and $c_t>0$.
%for some $t\in [T]$, i.e., $d_i\ne t$ for all $i\in [n]$, then we can eliminate this $t$ by redefining $t:= \min_{j: d_j> t} d_j$. Thus, we can without loss of generality assume that $\mathcal{I}(t)\ne \emptyset$ for all $t\in [T]$, which also implies that $T\le n$. 
We further note that the decision variables $x_i$'s in~\eqref{MPBKP} are binary, but if we relax this to any nonnegative integers, the problem becomes the so-called multiperiod knapsack problem (MPKP) as in~\cite{faaland1981multiperiod}. %As we will see in the next subsection, MPBKP can be viewed as a special case of MPBKP-S, and thus is not further discussed in this paper. 
%\begin{comment}
Our first main result is the following theorem on MPBKP.
\begin{theorem}\label{mainthm1}
	An FPTAS exists for MPBKP. Specifically, there exists a deterministic algorithm that achieves $(1+\epsilon)$-approximation in $\tilde{\mathcal{O}}\left(\min\left\{n+\frac{T^{3.25}}{\epsilon^{2.25}},n+\frac{T^{2}}{\epsilon^{3}},\frac{nT}{\epsilon^2},\frac{n^2}{\epsilon}\right\}\right)$.
\end{theorem}
%\end{comment}
As we will see shortly, MPBKP can be viewed as a special case of MPBKP-S. In Section~\ref{sec:MPBKP}, we will provide an approximation algorithm for MPBKP with runtime $\tilde{\mathcal{O}}\left(n+\frac{T^{3.25}}{\epsilon^{2.25}}\right)$. An alternative algorithm with runtime $\tilde{\Ocal}\left(n+\frac{T^{2}}{\epsilon^{3}}\right)$ is provided in Appendix~\ref{appT2}. In Section~\ref{sec:approx2}, we will provide an approximation algorithm for MPBKP-S with runtime $\tilde{O}\left(\frac{nT}{\epsilon^2}\right)$, which is also applicable to MPBKP.

\subsection{Multiperiod binary Knapsack problem with soft capacity constraints (MPBKP-S)}
In MPBKP-S, the capacity constraints in~\eqref{MPBKP} no longer exist, i.e., the total size of selected items at each time step is allowed to be greater than the total capacity up to that time, however, there is a penalty rate $B_t\in\NN$ for each unit of overflow at period $t$. We assume that $B_t>\max_{i\in[n]:d_i\le t}\frac{r_i}{q_i}$ to avoid trivial cases (any item with $\frac{r_i}{q_i}\ge B_t$ and $d_i\le t$ will always be added to generate more profit). In the IP form, MPBKP-S can be written as 
\begin{subequations}\label{MPBKP-S3}
	\begin{align}
	&\max_{x,y} \sum_{i\in[n]}r_ix_i - \sum_{t=1}^TB_ty_t\\
%	&\text{s.t. } \sum_{i\in\Ical(1)\cup\cdots\cup\Ical(t)}q_ix_i - \sum_{s=1}^ty_s \le \sum_{s=1}^ta_t = c_t,\quad \forall t: 1\le t\le T\\
	&\text{s.t. } \sum_{i\in\Ical(1)\cup\cdots\cup\Ical(t)}q_ix_i - \sum_{s=1}^ty_s \le c_t,\quad \forall t: 1\le t\le T\\
	&\qquad x_i\in\{0,1\},\quad y_t\ge 0,
	\end{align}
\end{subequations}
where the  decision variables $y_t, t=1,\ldots, T$ represent the units of overflow at time~$t$, and {  $c_t-c_{t-1}$} is the incremental capacity at time~$t$. The objective is to choose a subset of the $n$ items to maximize the total profit, which is the sum of the rewards of the selected items  minus the sum of penalty paid at each period, and the constraints enforce that the total size of accepted items by the end of each period must not exceed the sum of the cumulative capacity and the units of overflow. Our second main result is the following theorem on MPBKP-S.
\begin{theorem}\label{mainthm2}
	An FPTAS exists for MPBKP-S. Specifically, there exists an algorithm which achieves $(1+\epsilon)$-approximation in  ${\mathcal{O}}\left(\frac{n\log n}{\epsilon}\cdot \min\left\{\frac{T}{\epsilon}, n\right\}\right)$.
\end{theorem}

In section~\ref{sec:approx2} we will present an approximation algorithm for solving MPBKP-S with time complexity $\mathcal{O}\left(\frac{nT\log n}{\epsilon^2}\right)$. An alternative FPTAS with runtime $\Ocal\left(\frac{n^2}{\epsilon}\right)$ is provided in Appendix~\ref{simple-MPBKP-S}. 
For the ease of presentation, our algorithms and analysis are presented for the case $B_t=B$, but they can be generalized to the heterogeneous $\{B_1,\ldots,B_T\}$ in a straightforward manner. It is worth noting that the algorithm for MPBKP that we introduce in section~\ref{sec:MPBKP} does not extend to MPBKP-S, and we will make this clear in the beginning of section~\ref{sec:approx2}.

{ 
\subsection{Multiperiod Binary Knapsack Problem with Soft Stochastic Capacity Constraints (MPBKP-SS)}
The MPBKP-SS formulation is similar to~\eqref{MPBKP-S3}, except that the vector of knapsack sizes $(c_1, \ldots, c_T)$ follows some arbitrary joint distribution given to the algorithm as the set of possible sample path (realization) of knapsack sizes and the probability of each sample path. We use $\omega$ to index sample paths which we denote by $\{c_t(\omega)\}$, $p(\omega)$ as the probability of sample path $\omega$, and $\Omega$ as the set of possible sample paths. The goal is to pick a subset of items before the realization of $\omega$ so as to maximize the expected total profit, which is the sum of the rewards of the selected items deducted by the total (expected) penalty. For a sample $\omega\in\Omega$ let $y_t(\omega)$ be the overflow at time $t$. Then, we can write the problem in IP form as:
\begin{subequations}\label{MPBKP-SS}
\begin{align}
&\max_{x,y} \sum_{i\in[n]}r_ix_i - \mathbb{E}_\omega\left[B_t\cdot \sum_{t=1}^Ty_t(\omega)\right]\\
&\text{s.t. } \sum_{i\in\Ical(1)\cup\cdots\cup\Ical(t)}q_ix_i - \sum_{s=1}^ty_s(\omega) \le   c_t(\omega),\quad \forall \omega\in\Omega, 1\le t\le T\\
&\qquad x_i\in\{0,1\},\quad y_t\ge 0
\end{align}
\end{subequations}

Our third main result is the following theorem on MPBKP-SS, which asserts a greedy algorithm for the special case when all items are of the same size. Details will be provided in Section~\ref{sec:unit-MPBKP-SS}.
\begin{theorem}\label{mainthm3}
	If $q_i=q$ for all $i\in[n]$, then there exists a greedy algorithm that achieves $2$-approximation for MPBKP-SS in $\Ocal\left(n^2T|\Omega|\right)$.
\end{theorem}

We further note that both MPBKP-S and MPBKP-SS are special cases of non-monotone submodular maximization which is \emph{not} non-negative, for which not many general approximations are known. In that sense, studying these problems would be an interesting direction to develop techniques for it.
}

\section{FPTAS for MPBKP}\label{sec:MPBKP}
In this section, we provide an FPTAS for the MPBKP with time complexity $\tilde{\mathcal{O}}\left(n+\frac{T^{3.25}}{\epsilon^{2.25}}\right)$. We will apply the ``functional approach'' as used in~\cite{chan:OASIcs:2018:8299}. The main idea is to use the results on function approximations~\citep{chan:OASIcs:2018:8299,jin:LIPIcs:2019:10652} as building blocks -- for each period we approximate one function that gives, for every choice of available capacity, the maximum reward obtainable by selecting items in that period. We then combine ``truncated'' version of these functions using $(\max,+)$-convolution. This idea, despite its simplicity, allows us to obtain an FPTAS for MPBKP. Such a result should not be taken as granted -- as we will see in the next section, this method does not apply for MPBKP-S, even though it is just a slight generalization of MPBKP. 

We begin with some preliminary definitions and notations. For a given set of item rewards and sizes, $\Ical = \{(r_1,q_1), \ldots, (r_{n'}, q_{n'})\}$, define the function
\begin{align}\label{eqn:func}
f_\Ical(c) := \max_{x_1,\ldots,x_{n'}}\left\{\sum_{i\in\Ical}r_ix_i\ :\ \sum_{i\in\Ical}q_ix_i\le c, \ x_1,\ldots,x_{n'}\in\{0,1\}\right\}
\end{align}
for all $c\ge 0$, and $f_\Ical(c) := -\infty$ for $c<0$. The function $f_\Ical$ is a nondecreasing step function, and the number of steps is called the \emph{complexity} of that function. Further, for any $\Ical = \Ical_1\sqcup \Ical_2$, i.e., $\Ical$ being a disjoint union of $\Ical_1$ and $\Ical_2$, we have that $f_\Ical = f_{\Ical_1}\oplus f_{\Ical_2}$, where $\oplus$ denotes the $(\max,+)$-\emph{convolution}: $(f\oplus g)(c) = \max_{c'\in\mathbb{R}}\left(f(c')+g(c-c')\right)$.

We define the \emph{truncated function} $f_\Ical^{c'}$ as follows:
\begin{align}
f_\Ical^{c'}(c) = \begin{cases}
f_\Ical(c) &  c\le c',\\
-\infty & c>c'.
\end{cases}
\end{align}
Recall that we denote the set of items with deadline $t$ by $\Ical(t)$. We next define the function $f_t$ as follows: 
\begin{align}\label{eqn:ft}
f_t := \begin{cases}
f_{\Ical(1)}^{c_1} & t=1,\\
\left(f_{t-1}\oplus f_{\Ical(t)}\right)^{c_t} & t\ge 2.
\end{cases}
\end{align}
%In other words, $f_t$'s are defined recursively: for $t=1$, let $f_1 := f_{\Ical(1)}^{c_1}$; for $t\ge 2$, we define $f_t = \left(f_{t-1}\oplus f_{\Ical(t)}\right)^{c_t}$.
%, we write $f_t$ for $f_{\Ical(t)}$. 
In words, each function value of $f_t(c)$ corresponds to a feasible, in fact an optimal, solution $x$ for items with deadline at most $t$ as the next proposition shows.
\begin{proposition}\label{prop:optimalfunc}
Let $x^*$ be the optimal solution for MPBKP~\eqref{MPBKP}. We have that
the optimal value of~\eqref{MPBKP}, $\sum_{i\in[n]}r_ix_i^*$, satisfies
$
\sum_{i\in[n]}r_ix_i^*  = f_T(c_T).
$
\end{proposition}

Proposition~\ref{prop:optimalfunc} implies that, to obtain an approximately optimal solution for MPBKP~\eqref{MPBKP}, it is sufficient to have a good approximation for the function
\begin{align}
f_T = \left(\cdots\left(\left(f_{\Ical(1)}^{c_1}\oplus f_{\Ical(2)}\right)^{c_2}\oplus f_{\Ical(3)}\right)^{c_3}\cdots\oplus f_{\Ical(T)}\right)^{c_{T}}.
\end{align}

We say that a function $\tilde{f}$ approximates the nonnegative function $f$ with factor $1+\epsilon$ if $\tilde{f}(c)\le f(c)\le (1+\epsilon)\tilde{f}(c)$ for all $c\in\mathbb{R}$. It should be clear that if $\tilde{f}$ approximates $f$ with factor $1+\epsilon$ and $\tilde{g}$ approximates $g$ with factor $1+\epsilon$, then $\tilde{f}\oplus\tilde{g}$ approximates $f\oplus g$ with factor $1+\epsilon$. We then introduce the following result from~\cite{jin:LIPIcs:2019:10652} for 0-1 Knapsack problem. 
\begin{lemma}[\cite{jin:LIPIcs:2019:10652}]\label{lem:01}
Given a set $\Ical=\{(r_1,q_1),\ldots,(r_n,q_n)\}$, we can obtain $\tilde{f}_{\Ical}$ that approximates $f_\Ical$ (defined in~\eqref{eqn:func}) with factor $1+\epsilon$ and complexity $\tilde{O}\left(\frac{1}{\epsilon}\right)$ in $\tilde{O}\left(n+\left(1/\epsilon\right)^{2.25}\right)$.
\end{lemma}
%Suppose that $\tilde{f}_\Ical$ has complexity $l$, then, we denote by $(C_k,R_k)$ as the ``steps'' of function $\tilde{f}_\Ical$, i.e., for $k=1,\ldots,l$, we have that $\tilde{f}_\Ical(C_k) = R_k$ and $C_k = \min_{\tilde{f}_{\Ical}(c)=R_k}c$. 
With the above lemma, we present Algorithm~\ref{alg:MPBKP} for MPBKP.
\begin{algorithm}[ht]
\footnotesize
\caption{FPTAS for MPBKP}
\label{alg:MPBKP}
\algsetblock[Name]{Parameters}{}{0}{}
\algsetblock[Name]{Initialize}{}{0}{}
\algsetblock[Name]{Define}{}{0}{}
\begin{algorithmic}[1]
	\Statex \textbf{Input:} $\epsilon, [n], c_1,\ldots, c_T$  \Comment {Set of items to be packed, cumulative capacities}
	\Statex \textbf{Output:} $\tilde{f}_t$ \Comment Approximation of function $f_t$
	\State Discard all items with $r_i\le \frac{\epsilon}{n}\max_jr_j$ and relabel the items 
	\State $r_0\gets \min_ir_i$ \Comment Lower bound of solution value
	\State $m\gets \left\lceil\log_{1+\epsilon}\frac{n^2}{\epsilon}\right\rceil$ \Comment number of distinct rewards to be considered, each in the form $r_0\cdot(1+\epsilon)^k$
	%		\State Initialize $\tilde{A}(0,r) = \begin{cases}
	%		0 & r = 0,\\
	%		-\infty & r > 0.
	%		\end{cases}$
	\State Obtain $\tilde{f}_{\Ical(1)}$ that approximates $f_{\Ical(1)}$ with factor $(1+\epsilon)$ using Lemma~\ref{lem:01}
	\State $\tilde{f}_1:= \tilde{f}_{\Ical(1)}^{c_1}$ \Comment $\tilde{f}_1$ has complexity at most $m=\tilde{\mathcal{O}}\left(\frac{1}{\epsilon}\right)$
	\For {$t=2,\ldots, T$}
	\State Obtain $\tilde{f}_{\Ical(t)}$ that approximates $f_{\Ical(t)}$ with factor $(1+\epsilon)$ using Lemma~\ref{lem:01}
	\State $l\gets$ complexity of $\tilde{f}_{\Ical(t)}$ \Comment $l=\tilde{\mathcal{O}}\left(\frac{1}{\epsilon}\right)$
	\State Compute (all breakpoints and their values of) $\hat{f}_{t}:= \left(\tilde{f}_{t-1}\oplus \tilde{f}_{\Ical(t)}\right)^{c_t}$, taking $m\cdot l$ time
\Comment $\hat{f}_t$ has complexity $\tilde{\mathcal{O}}\left(\frac{1}{\epsilon^2}\right)$
	\State $\tilde{f}_t := r_0\cdot (1+\epsilon)^{\left\lfloor\log_{1+\epsilon}\left(\frac{\hat{f}_t}{r_0}\right)\right\rfloor}$ \Comment Round $\hat{f}_t$ down to the nearest $r_0\cdot (1+\epsilon)^k$. Now $\tilde{f}_t$ has complexity at most $m=\tilde{\mathcal{O}}\left(\frac{1}{\epsilon}\right)$
	\EndFor
\end{algorithmic}
\end{algorithm}

We now describe the intuition behind Algorithm~\ref{alg:MPBKP}. We first discard all items with reward $r_i\le \frac{\epsilon}{n}\max_jr_j$. The maximum we could lose is $n\cdot \frac{\epsilon}{n}\max_jr_j = \epsilon\max_jr_j$, which is at most $\epsilon$ fraction of the optimal value. We next obtain all $\tilde{f}_{\Ical(t)}$, for all $t=1,\ldots, T$, that approximate $f_{\Ical(t)}$ (as defined in~\eqref{eqn:func}) within a $(1+\epsilon)$ factor. These functions $\tilde{f}_{\Ical(t)}$ have complexity $\tilde{\Ocal}\left(\frac{1}{\epsilon}\right)$. We start with combining the functions of period~$1$ and period~$2$ using $(\max,+)$-convolution. To enforce the constraint that the total size of selected items in period~$1$ does not exceed the capacity of period~$1$, we truncate $\tilde{f}_{\Ical(1)}$ by $c_1$ (so that any solution using more capacity in period~$1$ results in $-\infty$ reward) and do the convolution on the truncated function $\tilde{f}_1$. Since both functions are step functions with complexity $\tilde{\Ocal}\left(\frac{1}{\epsilon}\right)$, the $(\max,+)$ convolution can be done in time $\Ocal\left(\frac{1}{\epsilon^2}\right)$. The resulting function $\hat{f}_2$ would have complexity $\Ocal\left(\frac{1}{\epsilon^2}\right)$. To avoid inflating the complexity throughout different periods (which increases computation complexity), the function $\hat{f}_2$ is rounded down to the nearest $r_0\cdot (1+\epsilon)^k$, where $r_0:=\min_jr_j$ and $k$ is some nonnegative integer. Note that $r_0$ is a lower bound of any solution value. After discarding small-reward items, we have that $\frac{\max_jr_j}{r_0}\le \frac{n}{\epsilon}$, which implies that $n\max_jr_j = \frac{n^2}{\epsilon}r_0$ is an upper bound for the optimal solution value. Therefore, after rounding down the function values of $\hat{f}_2$ and obtaining $\tilde{f}_2$, there are at most $\log_{1+\epsilon}\frac{n^2}{\epsilon}\approx \frac{1}{\epsilon}\log\frac{n^2}{\epsilon}$ different values on $\tilde{f}_2$. Now we have brought down the complexity of $\tilde{f}_2$ again to $\tilde{\Ocal}\left(\frac{1}{\epsilon}\right)$, at an additional $(1+\epsilon)$ factor loss in the approximation error. We then move to period~$3$ and continue this pattern of $(\max,+)$-convolution, truncation, and rounding down. In the end when we reach period $T$, $\tilde{f}_T$ will only contain feasible solutions to~\eqref{MPBKP}, and approximate $f_T$ with total approximation factor of $(1+\epsilon)^T\approx (1+T\epsilon)$. Formally, we have the following lemma which shows the approximation factor of $\tilde{f}_t$ for $f_t$.
\begin{lemma}\label{lem:fapprox}
Let $\tilde{f}_t$ be the functions obtained from Algorithm~\ref{alg:MPBKP}, and let $f_t$ be defined as in~\eqref{eqn:ft}. Then, $\tilde{f}_t$ approximates $f_t$ with factor $(1+\epsilon)^t$, i.e., $\tilde{f}_t(c)\le f_t(c)\le (1+\epsilon)^t\tilde{f}_t(c)$ for all $0\le c\le c_t$.
\end{lemma}

Lemma~\ref{lem:fapprox} and Proposition~\ref{prop:optimalfunc} together imply that $\tilde{f}_T(c_T)$, obtained from Algorithm~\ref{alg:MPBKP}, approximates the optimal value of MPBKP~\eqref{MPBKP} by a factor of $(1+\epsilon)^T \approx (1+T\epsilon)$. In Algorithm~\ref{alg:MPBKP}, obtaining $\tilde{f}_{\Ical(t)}$ for all $t=1,\ldots,T$ takes time $\tilde{O}\left(n+{T}/{\epsilon^{2.25}}\right)$; computing the $(\max,+)$-convolution on $\tilde{f}_{t-1}\oplus \tilde{f}_{\Ical(t)}$ for all $t$ take time $T\cdot m\cdot l = \tilde{O}\left(T/\epsilon^2\right)$. Therefore, Algorithm~\ref{alg:MPBKP} has runtime $\tilde{O}\left(n+T/\epsilon^{2.25}\right)$. As a result, we have the following proposition. 
\begin{proposition}
Taking $\epsilon' = T\epsilon$, Algorithm~\ref{alg:MPBKP} achieves $(1+\epsilon')$-approximation for MPBKP in $\tilde{O}\left(n+\frac{T^{3.25}}{{\epsilon'}^{2.25}}\right)$.
\end{proposition}

\section{FPTAS for MPBKP-S}\label{sec:approx2}

In this section, we provide an FPTAS for the MPBKP-S with time complexity $\mathcal{O}\left(\frac{Tn\log n}{\epsilon^2}\right)$. An alternative FPTAS with time complexity $\mathcal{O}\left(\frac{n^2\log n}{\epsilon}\right)$ is provided in Appendix~\ref{simple-MPBKP-S}. Combining the two, we show that our algorithms achieve $(1+\epsilon)$ approximation ratio in time $\mathcal{O}\left(\frac{n\log n}{\epsilon}\cdot\min\left\{\frac{T}{\epsilon},n\right\}\right)$, which proves Theorem~\ref{mainthm2}. 
We should note that the algorithm in the previous section does not apply here: we could similarly define a function which gives the maximum profit ($=$reward$-$penalty) under a given capacity constraint, but the main obstacle is on the $(\max,+)$-convolution because profit does not ``add up''. In other words, the total profit we earn by selecting items in the set $\Scal_1\cup\Scal_2$ is not the sum of the profits we earned by selecting $\Scal_1$ and $\Scal_2$ separately. For this reason, we can no longer rely on the techniques used in function approximation and $(\max,+)$-convolution as in~\cite{chan:OASIcs:2018:8299,jin:LIPIcs:2019:10652}. Instead, our main idea is motivated by the techniques that originated from earlier papers~\citep{ibarra1975fast,lawler1979fast}, but adapting their technique to MPBKP-S requires significant modifications as we show in this section. %To a large extent the algorithms and proofs follow the structure of Section~\ref{sec:approx} and for succinctness we only emphasize the modifications necessary. 
We restrict our presentation to the case $B_t =B$ for readability, but our algorithms and analysis generalize in a straightforward manner when the penalties for buying capacity are heterogeneous $\{B_1, \ldots, B_T\}$ (by replacing $B$ with $\min_{\tau\le t}B_\tau$ in the calculations of profit/penalty at period $t$ on line 7 of Algorithm~\ref{alg:FPTAS_nTlogn_large2}).

%The main idea is motivated by the technique that originated from~\cite{ibarra1975fast}, but adapting their technique to MPBKP requires significant modifications as we show in this Section. 

%\subsubsection*{Preparation}
\noindent \textbf{Preliminaries:} We first introduce some notation. From now on, let $\Rcal(\Scal):=\sum_{i\in\Scal}r_i$. The optimal solution set to~\eqref{MPBKP-S3} is denoted by $\mathcal{S}^*$. The total profit earned can be expressed as a function of the solution set $\mathcal{S}$:
\begin{align}
\Pcal(\mathcal{S}) = \Rcal(\Scal) - B\cdot \sum_{t=1}^T\left[\sum_{j\in \Scal\cap\mathcal{I}(t)}q_j-\max\left\{c_t-\sum_{j\in \Scal , d_j \leq t }q_j,\ c_t-c_{t-1}\right\}\right]^+.
\end{align}
Let $p_i$ be the profit of item $i$, which is defined as the profit earned if we select only $i$, i.e., $p_i = r_i-B \cdot \left(q_i-c_{d_i}\right)^+$. Without loss of generality, we assume that each item $i$ is by itself profitable, i.e., $p_i\ge 0$, so one profitable solution would be $\{i\}$. % This assumption is natural as otherwise there exists some item $i$ that will only bring down the profit if included in any solution, in which case we may simply discard that item when solving for the problem.
Let $P:=\max_{i}p_i$ and $\bar{P}:=\sum_{i\in[n]}p_i$. %Then since each item $i$ is profitable by itself, 
The following bounds on $\Pcal(\Scal^*)$  follow: 
\begin{align}\label{upperP}
P \leq \Pcal(\mathcal{S}^*)\le \bar{P} \leq nP.
\end{align}

%\subsubsection*{Partition of items}
\noindent \textbf{Partition of items:} We partition the set of items $[n]$ into two sets: a set of ``large'' items $\Ical_L$ and a set of ``small'' items $\Ical_S$ such that we can bound the number of large items in any optimal solution. The main idea is to use dynamic programming to pick the large items in the solution, and a greedy heuristic for `padding' this partial solution with small items.
The criterion for small and large items is based on balancing the permissible error $\epsilon \Pcal(\Scal^*)$ equally in filling large items and filling small items. Instead of first packing all large items and then all small items, we consider items in the order of their deadlines, and for each deadline $t$, the large items are selected first and then the small items are selected greedily in order of their reward densities. As a result, the approximation error due to large items overall will be $\frac{1}{2}\epsilon \Pcal(\Scal^*)$ and the error due to the small items with each deadline will be $\frac{1}{2T}\epsilon \Pcal(\Scal^*)$. This gives a total approximation error of $\frac{1}{2}\epsilon \Pcal(\Scal^*) + T\cdot \frac{1}{2T}\epsilon \Pcal(\Scal^*) = \epsilon \Pcal(\Scal^*)$.
%\subsection{An alternative FPTAS: separation of items}\label{sec:alter2}
%As in Section~\ref{sec:alter}, we propose in this subsection another FPTAS with time complexity $\mathcal{O}\left(\frac{Tn\log n}{\epsilon^2}\right)$. The algorithms we propose will be similar to Algorithms~\ref{alg:FPTAS_nTlogn_large},~\ref{alg:FPTAS_nTlogn_small},~\ref{alg:FPTAS_nTlogn}, and~\ref{alg:FPTAS_enumerate}. 

Suppose that we can find some $P_0$ that satisfies~\eqref{P0}. \begin{align}\label{P0}
P_0\le \Pcal(\mathcal{S}^*)\le 2P_0.
\end{align}
Then, the set of items is partitioned as follows.
\begin{equation}\label{div2}
\begin{aligned}
\Ical_L := \left\{i\in[n]\mid p_i\ge \frac{1}{2T}\epsilon P_0\right\}; \qquad 
\Ical_S := \left\{i\in[n]\mid p_i< \frac{1}{2T}\epsilon P_0\right\}.
\end{aligned}
\end{equation}

This partition is computed in $\mathcal{O}(n)$ time and is not the dominant term in time complexity. Let $n_L = |\Ical_L|$ and $n_S=|\Ical_S|$, so that $n_L+n_S=n$. 
Further, let 
\begin{align*}
\mathcal{I}_L(t) := \left\{i\in \Scal_L\mid d_i = t\right\}, \quad  \mbox{and} \quad 
\mathcal{I}_S(t) := \left\{i\in \Scal_S\mid d_i = t\right\}
\end{align*}
denote the set of large and small items, respectively, with deadline $t$. 
We will assume that the items in $\Ical_L$ are indexed in non-decreasing order of their deadlines, i.e., $\forall i,j\in \Ical_L$ such that $j\ge i$, we have that $ d_i\le d_j$. Denote by $I_L(t)$ as the index of the last item with deadline $t$, i.e., $I_L(t):= \max_{i\in \Scal_L\cap \mathcal{I}_L(t)} i$. 
For each time $t$, we will also sort the small items in $\mathcal{I}_S(t)$ according to their reward densities, i.e., $\forall i<j$ and $i,j\in \mathcal{I}_S(t)$, $\frac{r_i}{q_i}\ge \frac{r_j}{q_j}$. This sorting only takes place once for each guess $P_0$, and does not affect our overall time complexity result. \\

\begin{algorithm}[ht]
\footnotesize
\caption{DP on large items for MPBKP-S}
\label{alg:FPTAS_nTlogn_large2}
\algsetblock[Name]{Parameters}{}{0}{}
\algsetblock[Name]{Initialize}{}{0}{}
\algsetblock[Name]{Define}{}{0}{}
\begin{algorithmic}[1]
	\Statex \textbf{Input:} $\ \Ical_L, \Delta c,$  \Comment Set of (large) items to be packed, additional capacity available for packing
	\Statex \hspace{0.35in}	$\widetilde{A}(p)$ for all $p = \left\{ 0, 1, \ldots,\left\lceil\frac{16T}{\epsilon^2}\right\rceil \right\} \cdot \kappa $ \Comment A set of partial solutions 
	\Statex \textbf{Output:} $\hat{A}(I_L,p)$ for all $p = \left\{ 0, 1, \ldots,\left\lceil\frac{16T}{\epsilon^2}\right\rceil \right\} \cdot \kappa $ \Comment Set of partial solutions after packing $\Ical_L$
	\State Initialize $\forall p \ : \ \hat{A}(0,p) := \widetilde{A}(p) + \Delta c$	
	\For {$i=1,\ldots, I_L$}
	\For {$p = \left\{ 0, 1, \ldots,\left\lceil\frac{16T}{\epsilon^2}\right\rceil \right\} \cdot \kappa $}
	\State $\hat{A}(i,p) := \hat{A}(i-1, p)$ 
	\Comment If reject item $i$
	\EndFor
	\For {$\bar{p} = \left\{ 0, 1, \ldots,\left\lceil\frac{16T}{\epsilon^2}\right\rceil \right\} \cdot \kappa $}
	\State ${p} = \bar{p} + \hat{r}_i - \left\lceil B\left(q_i - {\color{black}\max\left\{0, \hat{A}(i-1,\bar{p})\right\}}\right)^+\right\rceil_{\kappa}$
	\State $\hat{A}(i, p) = \max\left\{ \hat{A}(i,p ), \hat{A}(i-1,\bar{p}) - q_i \right\}$		\Comment Accept $i$
	\EndFor
	\For{$p = \left\{\left\lceil\frac{16T}{\epsilon^2}\right\rceil,\left\lceil\frac{16T}{\epsilon^2}\right\rceil-1,\ldots,1  \right\}\cdot \kappa$}
	\vspace{0.1cm}
	\If {$\hat{A}(i,p-\kappa)<\hat{A}(i,p)$}
	\vspace{0.1cm}
	\State $\hat{A}(i,p-\kappa) = \hat{A}(i,p)$
	\EndIf
	\EndFor
	\EndFor
\end{algorithmic}
\end{algorithm}

\begin{algorithm}[h]
\footnotesize
\caption{Greedy on small items for MPBKP-S}
\label{alg:FPTAS_nTlogn_small2}
\algsetblock[Name]{Parameters}{}{0}{}
\algsetblock[Name]{Initialize}{}{0}{}
\algsetblock[Name]{Define}{}{0}{}
\begin{algorithmic}[1]
	\Statex \textbf{Input:} $\ \Ical_S$, $\hat{A}(p)$ for all $p = \left\{ 0, 1, \ldots,\left\lceil\frac{16T}{\epsilon^2}\right\rceil \right\} \cdot \kappa $.  \Comment Set of (small) items to be packed, a set of partial solutions
	\Statex \textbf{Output:} $\widetilde{A}(p)$ for all $p = \left\{ 0, 1, \ldots,\left\lceil\frac{16T}{\epsilon^2}\right\rceil \right\} \cdot \kappa $ \Comment Set of partial solutions after packing $\Ical_S$
	\State Initialize $\forall p \ : \ \widetilde{A}(p) = \hat{A}(p)$ 
	\For {$\bar{p}=\left\{ 0, 1, \ldots,\left\lceil\frac{16T}{\epsilon^2}\right\rceil \right\} \cdot \kappa$}
	\Statex  \texttt{// Filter out small items with size larger than $\hat{A}(p)$ } 
	%\State Sort the items in $\mathcal{I}_S(t)$ in decreasing order of reward density $r_i/q_i$
	\State $\widetilde{\mathcal{I}}_S\gets\emptyset$
	\For {$i\in \mathcal{I}_S$}
	\If {$\hat{A}(\bar{p}) \geq q_i$}
	%\State $\Delta p_{i} = r_{i} - B\cdot \left(q_{i}-\hat{A}(I_L(t),p)\right)$
	\State $\widetilde{\mathcal{I}}_S\gets \widetilde{\mathcal{I}}_S\cup \{i\}$
	\EndIf
	\EndFor

	\State $\tilde{\Rcal}_{0'} = 0, \tilde{q}_{0'} = 0$, and relabel the items in $\widetilde{\mathcal{I}}_S$ as $\left\{1',\ldots,|\widetilde{\mathcal{I}}_S|'\right\}$ (in decreasing order of reward density)
	\For {$i' = 1',\ldots, |\widetilde{\mathcal{I}}_S|'$}
	\State  $\tilde{\Rcal}_{i'} = \tilde{\Rcal}_{(i-1)'} + r_{i'}$
	\State  $\tilde{q}_{i'} = \tilde{q}_{(i-1)'} + q_{i'}$
	\EndFor
	
	\State {\texttt{// Add small items using Greedy algorithm}}
	\For {$i' = 1',\ldots, |\widetilde{\mathcal{I}}_S|'$}
	\If {$\tilde{q}_{i'} \leq \hat{A}(\bar{p})$}
	\State $p = \left\lfloor \bar{p} + \tilde{\Rcal}_{i'}\right\rfloor_\kappa$
	\State $\widetilde{A}({p}) = \max\left\{\widetilde{A}(p) , \hat{A}(\bar{p}) - \tilde{q}_{i'}  \right\}$
	\EndIf
	\EndFor
	\EndFor 
\end{algorithmic}
\end{algorithm}

\begin{algorithm}[ht]
\footnotesize
\caption{DP on large items and Greedy on small items for MPBKP-S}
\label{alg:FPTAS_nTlogn2}
\algsetblock[Name]{Parameters}{}{0}{}
\algsetblock[Name]{Initialize}{}{0}{}
\algsetblock[Name]{Define}{}{0}{}
\begin{algorithmic}[1]
	\Define \ $\kappa = \frac{\epsilon^2 P_0 }{8T}$
	\Define \ $\hat{r}_i = \floor{ {r_i} }_\kappa$ \Comment Round down reward
	\Statex  \texttt{// $\widetilde{A}_t(p)=$ leftover capacity for the algorithm's partial solution when earning (rounded) profit $p$ using items with deadlines at most $t$ (small and large)} %and rounded down supply
	\Statex  \texttt{// $\hat{A}_t(p)=$ capacity left for the algorithm's partial solution when earning (rounded) profit $p$ by selecting large items in $\mathcal{I}_L(t)$ with rounded down rewards $\hat{r}$, given the partial solutions $\widetilde{A}_{t-1}(p)$}
	\State Initialize $\hat{A}(0,p) = \widetilde{A}_0(p) = \begin{cases}
	0 & p = 0,\\
	-\infty & p > 0.
	\end{cases}$
	\For {$t=1,\ldots, T$}
	\State Run Algorithm~\ref{alg:FPTAS_nTlogn_large2} with $\Ical_L = \Ical_L(t), \Delta c = c_t - c_{t-1}$, and $\widetilde{A}(p)=\widetilde{A}_{t-1}(p)$ for all $p = \left\{ 0, 1, \ldots,\left\lceil\frac{16T}{\epsilon^2}\right\rceil \right\} \cdot \kappa $, and obtain $\hat{A}_t(p):= \hat{A}(I_L,p)$ for all $p$.
	\State Run Algorithm~\ref{alg:FPTAS_nTlogn_small2} with $\Ical_S=\Ical_S(t)$ and $\hat{A}(p) = \hat{A}(I_L(t),p)$ for all $p = \left\{ 0, 1, \ldots,\left\lceil\frac{16T}{\epsilon^2}\right\rceil \right\} \cdot \kappa $, and obtain $\widetilde{A}_{t}(p):=\widetilde{A}(p)$ for all $p$.
	\EndFor
\end{algorithmic}
\end{algorithm}

\begin{algorithm}[ht]
\footnotesize
\caption{FPTAS for MPBKP-S in $\mathcal{O}(Tn\log n/\epsilon^2)$}
\label{alg:SC_FPTAS}
\algsetblock[Name]{Parameters}{}{0}{}
\algsetblock[Name]{Initialize}{}{0}{}
\algsetblock[Name]{Define}{}{0}{}
\begin{algorithmic}[1]
	\State $P_0\gets {\bar{P}}$
	\State $p^*\gets 0$
	\While {$p^*<(1-\epsilon)P_0$}
	\vspace{0.1cm}
	\State $P_0\gets \frac{P_0}{2}$
	\vspace{0.1cm}
	\State	Run Algorithm~\ref{alg:FPTAS_nTlogn2} with the current $P_0$.
	\State $p^*\gets \max_{\left\{\substack{p\in \left\{ 0,\ldots,\left\lceil\frac{16T}{\epsilon^2}\right\rceil \right\} \cdot \kappa\\ \widetilde{A}_T(p)> -\infty}\right\}}p$
	\EndWhile
\end{algorithmic}
\end{algorithm}
%\subsubsection*{Overview of the algorithm}

\noindent \textbf{Algorithm overview:} Our FPTAS algorithm is given in Algorithm~\ref{alg:SC_FPTAS} which uses a doubling trick to guess the value of $P_0$ satisfying \eqref{P0}, and for each guess uses Algorithm~\ref{alg:FPTAS_nTlogn2} as a subroutine. Algorithm~\ref{alg:FPTAS_nTlogn2} is the main  algorithm for MPBKP-S, which first selects the items with deadline~$1$, then the items with deadline~$2$, and so on. 
For each deadline $t$, we maintain two sets of partial solutions: the first, $\widetilde{A}_t(p)$, corresponds to an approximately optimal (in terms of leftover capacity carried forward to time $t+1$) subset of large and small items with deadline at most $t$ and some \emph{rounded profit} $p$ %(precise definition of rounded profit will be given in~\eqref{Ptilde})
; and the second $\hat{A}_t(p)$ corresponds to the optimal appending of large items with deadline $t$ to the approximately optimal set of solutions corresponding to $\widetilde{A}_{t-1}$.

Given $\widetilde{A}_{t-1}$, we first select large items from $\mathcal{I}_L(t)$ using dynamic programming to obtain $\hat{A}_t$, which is done in Algorithm~\ref{alg:FPTAS_nTlogn_large2}. In other words, {\it given} the partial solutions $\widetilde{A}_{t-1}(\bar{p})$ for all $\bar{p} \in \left\{ 0, 1, \ldots,\left\lceil\frac{16T}{\epsilon^2}\right\rceil \right\} \cdot \kappa$, $\hat{A}_t(p)$ is the maximum capacity left when earning \emph{rounded profit} (precise definition given in~\eqref{Ptilde}) $p$ by adding items in $\mathcal{I}_L(t)$. We then use a greedy heuristic to pick small items from $\mathcal{I}_S(t)$ to obtain $\widetilde{A}_t$, which is done in Algorithm~\ref{alg:FPTAS_nTlogn_small2}. Specifically, our goal in Algorithm~\ref{alg:FPTAS_nTlogn_small2} is to obtain the partial solutions $\widetilde{A}_t(\cdot)$ given the partial solutions $\hat{A}_t(\cdot)$ by packing the small items $\Ical_S(t)$. We initialize $\widetilde{A}_t(\bar{p})$ with $\hat{A}_t(\bar{p})$, and for each $\bar{p}$ we try to augment the solution corresponding to $\hat{A}_t(\bar{p})$ using a subset $\widetilde{\mathcal{I}}_S(t) \subseteq \Ical_S(t)$ defined as $$\widetilde{\mathcal{I}}_S(t):= \{i\in \mathcal{I}_S(t)\mid q_i\le \hat{A}_t(\bar{p})\}.$$  The small items in $\widetilde{\mathcal{I}}_S(t)$ are sorted according to their reward densities, and are added to the solution of $\hat{A}_t(\bar{p})$ one by one. After each addition of a small item, if the new total rounded reward is ${p}$, we compare the leftover capacity with current $\widetilde{A}_t({p})$, and update $\widetilde{A}_t({p})$ with the new solution if it has more leftover capacity. We continue this add-and-compare (and possibly update) until we reach the situation where adding the next small item overflows the available capacity.

Intuitively, for any amount of capacity available to be filled by small items, and a minimum increase in profit, the optimal solution either packs a single item from $\Ical_S(t) \setminus \widetilde{\Ical}_S(t)$ in which case the loss by ignoring items in this set is bounded by the maximum reward of any small item, or the optimal solution only contains items from $\widetilde{\Ical}_S(t)$ in which case the space used by this optimal set of items is lower bounded by the a fractional packing of the highest density items in  $\widetilde{I}_S(t)$. During Algorithm~\ref{alg:FPTAS_nTlogn_small2}, one of the solutions we would consider would be the integral items of this fractional solution, and lose at most $\frac{1}{2T}\epsilon P_0$ in profit, and obtain a solution with still smaller space used (more leftover capacity) than the fractional solution. Accumulation of these errors for $t$ periods then will give us the invariant: the partial solution $\widetilde{A}_t(p)$ obtained as above has more leftover capacity than any solution obtained by selecting items from $\cup_{t'=1}^t \mathcal{I}_L(t')$ with rounded rewards and rounded penalties, and items from $\cup_{t'=1}^t \mathcal{I}_S(t')$ with original (unrounded) rewards such that the rounded total profit is at least $p+\frac{1}{2T}\epsilon P_0t+\kappa t$.

%Recall that we have assumed $R_0\le \Rcal(\Scal^*)\le 2R_0$. To find such an $R_0$, we would again enumerate $R_0$ from $\bar{R}/2, \bar{R}/4, \bar{R}/8,\ldots$, and one of them must satisfy~\eqref{R0}. This is done in Algorithm~\ref{alg:FPTAS_enumerate}. 
Our main theorem for the approximation ratio for MPBKP follows.

\begin{theorem}[Partially restating Theorem~\ref{mainthm2}]\label{main:MPBKP-S}
Algorithm~\ref{alg:SC_FPTAS} is a fully polynomial approximation scheme for the MPBKP-S, which achieves $(1+\epsilon)$ approximation ratio with running time $\mathcal{O}\left(\frac{Tn\log n}{\epsilon^2}\right)$.
\end{theorem}

\section{A greedy algorithm for a special case of MPBKP-SS}\label{sec:unit-MPBKP-SS}
In this subsection, we consider the special case of MPBKP-SS when all items have the same size, i.e., $q_i=q,\forall i\in[n]$. We again only present for the case $B_t = B,\forall t\in[T]$. We note that in the deterministic problems (MPBKP or MPBKP-S), when items all have the same size, greedily adding items one by one in decreasing order of their rewards leads to the optimal solution. For MPBKP-SS, as the capacities are now stochastic, we wonder if there is any greedy algorithm performs well. We propose Algorithm~\ref{alg:unitq-greedybyprofit}, where we start with an empty set, and greedily insert the item that brings the maximum increment on expected profit, and we stop if adding any of the remaining items does not increase the expected profit.

\begin{algorithm}[h]
	\footnotesize
	\caption{Greedy algorithm according to profit change}
	\label{alg:unitq-greedybyprofit}
	\algsetblock[Name]{Parameters}{}{0}{}
	\algsetblock[Name]{Initialize}{}{0}{}
	\algsetblock[Name]{Define}{}{0}{}
	\begin{algorithmic}[1]
		\State $\Scal\gets \emptyset$
		\State $s\gets 1$
		\While {$s == 1$}
		\State $i^*\gets \argmax_{i \notin\Scal}\left\{\Pcal(\Scal\cup\{i\})-\Pcal(\Scal)\right\}$
		\If {$\Pcal(\Scal\cup\{i^*\})-\Pcal(\Scal)\ge 0$}
		\State $\Scal\gets \Scal\cup\{i^*\}$
		\Else 
		\State $s\gets 0$
		\EndIf
		\EndWhile
		\State $\Scal_{p}\gets \Scal$
		\State {\bf Return} $\Scal_p$
	\end{algorithmic}
\end{algorithm}

Let $\Scal^*$ be an optimal solution, i.e.,
%\begin{align}\label{obj:unitsize}
$\Scal^* \in \arg\max_{\Scal\subseteq [n]} \Pcal(\Scal) := \Rcal(\Scal) - B\cdot \Phi(\Scal)$,
%\end{align}
where $$\Phi(\Scal) := \mathbb{E}\left\{\sum_{t=1}^T\left[\sum_{j\in \mathcal{I}(t)\cap\Scal}q_j-\max_{0\leq t' < t}\left\{ c_t - c_{t'}-\sum_{j  \in \Scal : t'+1 \leq d_j \leq t-1} q_j\right\}\right]^+ \right\}$$ is the expected quantity of overflow on set $\Scal$, and let $\Scal_p$ be the set output by Algorithm~\ref{alg:unitq-greedybyprofit}.
Then, we have the following theorem. 
\begin{theorem}[Restating Theorem~\ref{mainthm3}]\label{thm:GRprofit}
	Algorithm~\ref{alg:unitq-greedybyprofit} achieves $2$-approximation factor for MPBKP-SS when items have the same size, i.e., $\Pcal(\Scal_p) \ge \frac{1}{2}\Pcal(\Scal^*)$  in $\Ocal\left(n^2T|\Omega|\right)$.
\end{theorem}

The proof of the $2$-approximation could be more nontrivial than one may think. The idea is to look at the greedy solution set $\Scal_p$ and the optimal solution set $\Scal^*$,  where we will use the dual to characterize the optimal solution on each sample path. By swapping each item in $\Scal_p$ to $\Scal^*$ in replacement of the same item or two other items, we construct a sequence of partial solutions of the greedy algorithm as well as modified optimal solution set, while maintaining the invariant that the profit of $\Scal^*$ is bounded by the sum of two times the profit of items in $\Scal_p$ swapped into $\Scal^*$ so far and the additional profit of remaining items in the modified optimal solution set. We leave the formal proof of Theorem~\ref{thm:GRprofit} to Appendix~\ref{appc-unit}.

\section{Comments and Future Directions}\label{sec:conc}
The current work represents to the best of our knowledge the first FPTAS  for the two multi-period variants of the classical knapsack problem. For MPBKP, we obtained the runtime $\tilde{\Ocal}\left(n+(T^{3.25}/\epsilon^{2.25})\right)$. This was done via the function approximation approach, where a function is approximated for each period. The runtime increases in $T$ since we conduct $T$ number of rounding downs, one after each $(\max,+)$-convolution. An alternative algorithm with runtime $\tilde{\Ocal}\left(n+\frac{T^{2}}{\epsilon^{3}}\right)$ is also provided in Appendix~\ref{appT2}. Note that the function we approximated is in the same form as used in the 0-1 knapsack problem~\citep{chan:OASIcs:2018:8299}. It is thus interesting to ask if we could instead directly approximate the following function:
$$
f_{\Ical}(c) = \max_{x}\left\{\sum_{i\in\Ical}r_ix_i\ :\ \sum_{i\in\cup_{t'=1}^t\Ical(t')}q_ix_i\le c_t,\forall t\in[T],\ x\in\{0,1\}^n\right\},
$$
where $\Ical = \cup_{t=1}^T\Ical(t)$ and $c=\{c_1,\ldots,c_T\}$ is a $T$-dimensional vector. Here we impose all $T$ constraints in the function. The hope is that, if the above function could be approximated, and if we could properly define the $(\max,+)$-convolution on $T$ dimensional vectors (and have a fairly easy computation of it), then we may get an algorithm that depends more mildly on~$T$.

For MPBKP-S and MPBKP-SS, there seems to be less we can do without further assumptions. One direction to explore is  parameterized approximation schemes: assuming that in the optimal solution, the total (expected) penalty is at most~$\beta$ fraction of the total reward. Then we may just focus on rewards. Our ongoing work suggests that an approximation factor of $\left(1+\frac{\epsilon}{1-\beta}\right)$ may be achieved in $\tilde{\Ocal}\left(n+(T^{3.25}/\epsilon^{2.25})\right)$ for MPBKP-S, and the same approximation factor in $\tilde{\Ocal}\left(n+\frac{1}{\epsilon^{T}}\right)$ for MPBKP-SS. 

We further note that the objective function for the three multiperiod variants are in fact submodular (but not non-negative, or monotone). Whether we can get a constant competitive solution in time $\widetilde{\Ocal}(n)$, using approaches in submodular function maximization, is also an intriguing open problem. 

Finally, motivated by applications, one natural extension that the authors are working on now is when there is a general non-decreasing cost function $\phi_t(\Delta c)$ for procuring capacity $\Delta c$ at time $t$, and the goal is to admit a profit maximizing set of items when the unused capacity can be carried forward. Another extension is when there is a bound on the leftover capacity that can be carried forward.

	{\small
		\begin{spacing}{1.2}
			\bibliographystyle{apalike}
			\bibliography{references}
		\end{spacing}
	}

	\begin{appendix}
		\section{Omitted Proofs}\label{Appb}

\subsection{Proofs for Section~\ref{sec:MPBKP}}
\begin{proof}[Proof of Proposition~\ref{prop:optimalfunc}]
	We show that the solution corresponding to $f_T(c)$ is optimal for $c_T=c$ among all solutions feasible to~\eqref{MPBKP}. We prove by induction on $T$. Base case is $T=1$, this reduces to 0-1 Knapsack problem, and by definition, the solution corresponding to $f_{\Ical(1)}(c)$ is the optimal feasible solution when the Knapsack capacity is $c$. For the induction step, assume that the solution of $f_{T-1}(c')$ is the optimal feasible solution to~\eqref{MPBKP} for the $T-1$ period problem and $c_{T-1}=c'$, we show that the solution corresponding to $f_{T}(c)$ is also the optimal feasible solution to~\eqref{MPBKP} for the $T$ period problem and $c_{T}=c$. 
	
	By definition, $$f_{T}(c) = \left(\left(f_{T-1}\oplus f_{\Ical(T)}\right)(c)\right)^{c_{T}} = \left(\max_{c'\in\mathbb{R}}\left(f_{T-1}(c')+f_{\Ical(T)}(c-c')\right)\right)^{c}.$$ We first show that $f_T(c)$ is at least the optimal value of~\eqref{MPBKP} when $c_T=c$. Suppose that, in the optimal solution of~\eqref{MPBKP}, the total size of accepted items up to time $T-1$ is $\hat{c}$ with $\hat{c}<c$, then the optimal value is $f_{T-1}(\hat{c})+f_{\Ical(T)}(c-\hat{c})$ since $f_{T-1}(\hat{c})$ is the maximum achievable reward with $c_{T-1}=\hat{c}$ (by induction assumption) and $f_{\Ical(T)}(c-\hat{c})$ is the maximum achievable reward using items from $\Ical(T)$ with space constraint $c-\hat{c}$. Thus, we have that the optimal value $f_{T-1}(\hat{c})+f_{\Ical(T)}(c-\hat{c})\le \left(\max_{c'\in\mathbb{R}}\left(f_{T-1}(c')+f_{\Ical(T)}(c-c')\right)\right)^c=f_T(c)$. 
	
	We next show the other direction: the optimal value of~\eqref{MPBKP} for the $T$ period problem with $c_T=c$ is at least $f_T(c)$. It suffices to show that every possible solution considered in $f_T(c)$ satisfies the feasibility constraints in~\eqref{MPBKP}. By induction assumption, every solution of $f_{T-1}(c')$ satisfies the constraints up to time $T-1$. When computing $f_T(c)$, we note that since $f_{T-1}(c')$ is a function truncated at $c_{T-1}$, which implies that $f_{T-1}(c')=-\infty$ for any $c'>c_{T-1}$. Therefore, any $c'>c_{T-1}$ must not be in the solution of $\max_{c'\in\mathbb{R}}\left(f_{T-1}(c')+f_{\Ical(T)}(c-c')\right)$. As a result, every solution of $f_T(c)$ is enforcing that ${c'}\le c_{T-1}$, and satisfies the feasibility constraints up to time $T$.
	
	Combining both directions, we conclude the induction step, and thus the proof of the proposition. 
\end{proof}

\begin{proof}[Proof of Lemma~\ref{lem:fapprox}]
By the construction of $\tilde{f}_t$, it should be clear that $\tilde{f}_t\le f_t$. We prove that $(1+\epsilon)^t\tilde{f}_t\ge f_t$ by induction on $t$. Base case is when $t=1$, we have that $(1+\epsilon)\tilde{f}_1=(1+\epsilon)\tilde{f}_{\Ical(1)}^{c_1} \ge f_{\Ical(1)}^{c_1}=f_1$, where the inequality follows from Lemma~\ref{lem:01}. As for the induction step, assume that $(1+\epsilon)^{t-1}\tilde{f}_{t-1}\ge f_{t-1}$, we show that $(1+\epsilon)^t\tilde{f}_t\ge f_t$. Again, by Lemma~\ref{lem:01} we have that
$$
(1+\epsilon)^{t-1}\tilde{f}_{\Ical(t)}\ge (1+\epsilon)\tilde{f}_{\Ical(t)}\ge f_{\Ical(t)}.
$$
Combined with the induction hypothesis, we have that
$$
(1+\epsilon)^{t-1}\left(\tilde{f}_{t-1}\oplus \tilde{f}_{\Ical(t)}\right) = \left((1+\epsilon)^{t-1}\tilde{f}_{t-1}\right) \oplus \left((1+\epsilon)^{t-1}\tilde{f}_{\Ical(t)}\right)\ge f_{t-1}\oplus f_{\Ical(t)}.
$$
Taking truncation on both sides, we have that $$
(1+\epsilon)^{t-1}\hat{f}_t=(1+\epsilon)^{t-1}\left(\tilde{f}_{t-1}\oplus \tilde{f}_{\Ical(t)}\right)^{c_t}\ge \left(f_{t-1}\oplus f_{\Ical(t)}\right)^{c_t}=f_t.
$$
Because of rounding down, we have that $(1+\epsilon)\tilde{f}_t\ge \hat{f}_t$. Therefore, 
$$
(1+\epsilon)^t\tilde{f}_t\ge (1+\epsilon)^{t-1}\hat{f}_t\ge f_t.
$$
This concludes the induction step, and thus the proof of the lemma.
\end{proof}

\subsection{Proofs for Section~\ref{sec:approx2}}

This section is devoted to the proof of Theorem~\ref{main:MPBKP-S}. 
%We will mimic the proof we did for Theorem~\ref{main:MPBKP} as in the previous section, with changes wherever necessary. 
To proceed, we first present the following result on Algorithm~\ref{alg:FPTAS_nTlogn_large2}.
\begin{lemma}\label{hat2}
Given a set of partial solutions with leftover capacities $\widetilde{A}({p})$ for all ${p}\in\left\{ 0, 1, \ldots,\left\lceil\frac{16T}{\epsilon^2}\right\rceil \right\} \cdot \kappa$, the additional capacity available for packing $\Delta c$, and the set of large items to be added $\Ical_L:=\{1,\ldots, I_L\}$, the output of Algorithm~\ref{alg:FPTAS_nTlogn_large2}, $\hat{A}(I_L,p)$, satisfies: 
\begin{align}\label{largeDP2}
\hat{A}(I_L,p) = \max_{\left\{\substack{ \Ical',\bar{p}\  : \ 
		\mathcal{I}'\subseteq \mathcal{I}_L\\
		\Delta \hat{\Pcal}(\Ical', \widetilde{A}(\bar{p})+\Delta c) \ge p-\bar{p}\\
		\bar{p}\in\left\{ 0, 1, \ldots,\left\lceil\frac{16T}{\epsilon^2}\right\rceil \right\} \cdot \kappa}\right\}} \widetilde{A}(\bar{p})+ \Delta c-\Qcal(\Ical'),\quad \forall p.
\end{align}
That is, $\hat{A}(I_L,p)$ is the maximum leftover capacity for any solution with (rounded) profit at least $p$ obtained by adding items in $\Ical_L$ to the solutions corresponding to $\widetilde{A}(\cdot)$.
\end{lemma}
%The proof of Lemma~\ref{hat2} proceeds in a similar manner as the proof of correctness of Algorithm~\ref{alg:FPTAS_SC_1}, and is left to Appendix~\ref{Appb}.
\begin{proof}[Proof of Lemma~\ref{hat2}]
	We will prove a more general result than~\eqref{largeDP2}, i.e.,
	\begin{align}\label{DPpf2}
	\hat{A}(i,p) = \max_{\left\{\substack{ \Ical',\bar{p}\  : \ 
			\mathcal{I}'\subseteq \{1,\ldots,i\}\\
			\Delta \hat{\Pcal}(\Ical', \widetilde{A}(\bar{p})+\Delta c) \ge p-\bar{p}\\
			\bar{p}\in\left\{ 0, 1, \ldots,\left\lceil\frac{16T}{\epsilon^2}\right\rceil \right\} \cdot \kappa}\right\}} \widetilde{A}(\bar{p})+ \Delta c-\Qcal(\Ical'),\quad \forall p
	\end{align}
	We prove this by induction. The base case ($i=0$) is vacuously true. Now we assume that \eqref{DPpf2} holds for all $p \in \left\{  0, 1, \ldots, \lceil 16T/\epsilon^2 \rceil \right\}  \kappa$ and for all $k \in [i-1]$. Consider some $p \in \left\{  0, 1, \ldots, \lceil 16T/\epsilon^2 \rceil \right\}  \kappa $, and let $\Ical^*$ be any set achieving the maximum in \eqref{DPpf2} so that $\hat{P}(\Ical^*) \ge p-\bar{p}$ for some $\bar{p}\in \left\{ 0, 1, \ldots,\left\lceil\frac{16T}{\epsilon^2}\right\rceil \right\} \cdot \kappa$. We will show that $\hat{A}(i,p)$ is at least the leftover capacity under solution $\Ical^*$ via case analysis:
	\begin{itemize}
		\item Case $i \notin \Ical^*$: In this case, the leftover capacity under $\Ical^*$ is the leftover capacity by $d_i$, which is the sum of leftover capacity in $\Ical^*$ by $d_{i-1}$ and $c_{d_i}-c_{d_{i-1}}$. By induction hypothesis, $\hat{A}(i-1,p)$ is no less than the leftover capacity of $\Ical^*$ by $d_{i-1}$, and therefore, by lines~$4$ and~$8$, $\hat{A}(i,p) \geq \hat{A}(i-1,p) + c_{d_i}-c_{d_{i-1}}$ which in turn is no less than the leftover capacity under $\Ical^*$ by $d_i$. By optimality of $\Ical^*$, all the inequalities must be equalities.
		\item Case $i \in \Ical^*$: Let $\Ical' = \Ical^* \setminus \{ i\}$, and let $p' = \hat{\Pcal}(\Ical')$ be its rounded profit. Then by induction hypothesis, $\hat{A}(i-1,p')$ is no less than the leftover capacity under $\Ical'$ by $d_{i-1}$. Further, by packing item $i$ in the solution corresponding to $\hat{A}(i-1,p')$, the change in profit is larger than by packing item $i$ in $\Ical'$ (the penalty is no less under $\Ical'$ since it has weakly smaller leftover capacity). Therefore, packing item $i$ in the solution corresponding to $\hat{A}(i-1,p')$ gives a solution with at least as large a rounded profit as $p$ and at least as much leftover capacity by $d_i$ as $\Ical^*$. Therefore, in turn $\hat{A}(i,p)$ is at least as much as the leftover capacity in $\Ical^*$. Since we assume $\Ical^*$ to have the largest leftover capacity with profit at least $p$, all the inequalities must be equalities.
	\end{itemize}
	This completes the induction step, and thus the proof of the lemma. 
\end{proof}

Next, we have the following Lemma as a preparation for our result on $\widetilde{A}(p)$ of Algorithm~\ref{alg:FPTAS_nTlogn_small2}.
\begin{lemma}\label{singleperiod2}
Given some capacity $c$ and a set of small items $\Ical_S$ with $p_{max}:=\max_{i\in \Ical_S}p_i$, let $\Scal^*$ be the profit-optimal subset, i.e., $\Scal^*=\arg\max_{
	\Scal\subseteq \Ical_S}\Pcal(\Scal)=\Rcal(\Scal)-B\left(\Qcal(\Scal)-c\right)^+$. Further, let $\widetilde{\Ical}_S := \{i\in \Ical_S\mid q_i\le c\}$ and relabel the items in $\widetilde{\Ical}_S$ as $\left\{1',\ldots,|\widetilde{\mathcal{I}}_S|'\right\}$ (in decreasing order of reward density $r_i/q_i$). Let $i'$ be such that $\sum_{j'=1'}^{i'}q_{j'}\le c$ and $\sum_{j'=1'}^{(i+1)'}q_{j'}> c$. Then, the solution $\Scal':=\{1',\ldots,i'\}$ satisfies
\begin{itemize}
	\item $\Qcal(\Scal') \le \Qcal(\Scal^*)$,
	\item $\Pcal(\Scal')\ge \Pcal(\Scal^*) - p_{max}$.
\end{itemize}
\end{lemma}
\begin{proof}[Proof of Lemma~\ref{singleperiod2}]
	The first item can be shown by contradiction. Suppose that to the contrary $\Qcal(\Scal') > \Qcal(\Scal^*)$, that is, $\Scal'$ uses more space than $\Scal^*$. Since the items in $\Scal'$ have the highest reward densities, it is in fact the optimal solution which uses space $\Qcal(\Scal')<c$. Since the optimal profit is non-decreasing in the capacity $c$, this violates optimality of $\Scal^*$.
	
	To see the second item, we look at two different cases. First, if $\Scal^*\cap \left(\Ical_S\setminus \widetilde{\Ical_S}\right)\ne \emptyset$, i.e., the optimal packing $\Scal^*$ includes some item $i^*$ with $q_{i^*}>c$, then, there should be only one item in $\Scal^*$, i.e., $\Scal^* = \{i^*\}$. In this case, $\Pcal(\Scal^*) = p_{i^*} = p_{max}$ and thus $\Pcal(\Scal')\ge \Pcal(\emptyset) = 0 = \Pcal(\Scal^*)-p_{max}$.
	
	Second, if $\Scal^*\cap \left(\Ical_S\setminus \widetilde{\Ical_S}\right)= \emptyset$, then $\Scal^* = \arg\max_{\Scal\subseteq \widetilde{I}_S}\Pcal(\Scal)$. Note that $\Pcal(\Scal^*)$ is upper bounded by the reward for the fractional packing: $ \Pcal(\Scal^*) \leq \Rcal_{LP} := \Rcal(\Scal') + r_{(i+1)'}\cdot \frac{ c - \Qcal(\Scal')}{q_{(i+1)'}} \leq \Rcal(\Scal') + r_{(i+1)'} = \Pcal(\Scal') + p_{(i+1)'}\le \Pcal(\Scal') + p_{max}$.
	
	In either cases, we conclude that $\Pcal(\Scal')\ge \Pcal(\Scal^*) - p_{max}$.
\end{proof}

Before presenting our result on $\widetilde{A}_t(p)$, we will need the following definitions. 
%For a solution $\Scal\subseteq [n]$ we define {\it the rounded profit} $\tilde{\Pcal}(\Scal)$ as,
%\begin{align}\label{ptilde}
%\tilde{\Pcal}(\Scal) := \hat{\Pcal}(\Scal\cap \Ical_L) + \sum_{t \in [T]} \left\lfloor \Pcal(\Scal\cap \Ical_S(t))\right\rfloor_\kappa. %= \kappa\left\lfloor\frac{\sum_{i\in \Scal\cap \Scal_L} \kappa \floor{ \frac{r_i}{\kappa} } + \sum_{j\in\Scal\cap \Scal_S} r_j}{\kappa}\right\rfloor
%\end{align}
For a solution $\Scal =  \Scal(1) \cup \Scal(2) \cup \cdots \cup \Scal(T)$ with $\Scal(t) = \Scal_L(t) \cup \Scal_S(t)$, denoting the items with deadline $t$ in $\Scal$, let the large items be indexed as $\Scal_L(t) = (i^{(t)}_1, \ldots, i^{(t)}_{L_t})$ in the order in which Algorithm~\ref{alg:FPTAS_nTlogn_large2} considers them, and the small items be indexed arbitrarily $\Scal_S(t) = \left(j^{(t)}_1, \ldots, j^{(t)}_{S_t}\right)$. Let $\Scal_L := \Scal_L(1) \cup \cdots \cup \Scal_L(T)$ and $\Scal_S := \Scal_S(1) \cup \cdots \cup \Scal_S(T)$ denote the large and small items in $\Scal$, respectively (this depends on the choice of $P_0$ but we suppress the dependence for brevity). We define the {\it rounded profit} of $\Scal$ as:

\begin{align} \label{Ptilde}
\nonumber
\tilde{\Pcal}(\Scal) &= \hat{\Rcal}(\Scal_L) - \sum_{t=1}^T  \sum_{k=1}^{L_t} \left\lceil  B  \left( \sum_{\ell \leq k} q_{i^{(t)}_{\ell}} - \max_{0 \leq t' < t}\left\{ c_t - c_{t'} - \sum_{t'+1 \leq \tau <  t} \Qcal(\Scal(\tau)) \right\} \right)^+  \right\rceil_{\kappa} \\
& \quad 	+ \sum_{t=1}^T \left\lfloor \Rcal(\Scal_S(t)) - B \left(  \Qcal(\Scal(t))  - \max_{0 \leq t' < t}\left\{ c_t - c_{t'} - \sum_{t'+1 \leq \tau <  t} \Qcal(\Scal(\tau)) \right\} \right)^+ \right\rfloor_\kappa.
\end{align}

That is, we add the rounded rewards of the large items, and for small items, we first group the small items by their deadlines, and for each deadline we round the sum of unrounded rewards of small item.
Further, let 
\[ \widetilde{C}_t(p) := \max_{\left\{\substack{
	\Scal \subseteq \bigcup_{t'=1}^t\mathcal{I}(t')\ :\  \tilde{\Pcal}(\Scal)\ge p}
\right\}} \max_{0 \leq t' < t}\left\{ c_t - c_{t'} - \sum_{t'+1 \leq \tau \leq  t} \Qcal(\Scal(\tau)) \right\} \]
denote the feasible partial solution with largest leftover capacity at time $t$ and rounded total profit at least $p$. Then, we have the following lemma.
\begin{lemma}\label{tilde2}
For any $t=1,\ldots,T$ and any $p'\in \left\{ 0, 1, \ldots,\left\lceil\frac{16T}{\epsilon^2}\right\rceil \right\} \cdot \kappa$, we have that $\widetilde{A}_t(p) \ge \widetilde{C}_t(p')$ for some $p\ge  p'-\frac{1}{2T}\epsilon P_0t - \kappa t \geq  r'- \frac{1}{2T}\epsilon (1-\epsilon/4) P_0t $.  %- \frac{\epsilon^2 R_0}{8T}t$.
That is, for any rounded total profit $p'$ by time $t$, there exists some partial solution $\widetilde{A}_t$ of Algorithm~\ref{alg:FPTAS_nTlogn2} which has at least as much leftover capacity at time $t$ the optimal solution $\widetilde{C}_t(p')$, and has rounded profit $p$ not too much smaller than $p'$.
\end{lemma}
%The proof of Lemma~\ref{tilde2} is similar to the proof of Lemma~\ref{tilde}, with the new rounded profit and the newly defined $\widetilde{A}_t(p),\widetilde{C}_t(p)$, etc. We refer to Appendix~\ref{Appb} for the complete proof.
\begin{proof}[Proof of Lemma~\ref{tilde2}]
	We prove by induction on $t$. Base case is when $t=1$. Let $\Scal'$ be the solution corresponding to $\widetilde{C}_1(p')$, i.e., $\Scal' := \arg\max_{\left\{\substack{\Scal \subseteq \mathcal{I}(1)\\
			\tilde{\Pcal}(\Scal)\ge p'}\right\}} c_1-\Qcal(\Scal)$, and let $\Scal'_L = \Scal'\cap \Ical_L$, $\Scal'_S = \Scal'\cap \Ical_S$. Then $\tilde{\Pcal}(\Scal'_L) = \hat{\Pcal}(\Scal_L')$. By Lemma~\ref{hat2}, $\hat{A}(I_L(1),\tilde{\Pcal}(\Scal_L'))$ is the maximum leftover capacity using items in $\mathcal{I}_L(1)$ earning rounded profit $\tilde{\Pcal}(\Scal_L')$. Thus, $\hat{A}_1(\tilde{\Pcal}(\Scal_L'))  = \hat{A}(I_L(1),\tilde{\Pcal}(\Scal_L'))\ge c_1-\Qcal(\Scal_L')$. Let  $\Scal''_L $ be the solution corresponding to $\hat{A}_1(\tilde{\Pcal}(\Scal_L'))$, and thus $\Qcal(\Scal_L'')\le \Qcal(\Scal_L')$. 
	Consider appending the partial solution $\Scal_L''$ using items from $\Ical_S(1)$. Let $\Scal_S''$ be the small item set obtained by adding small items greedily in their reward densities, subject to the constraint that $\Qcal(\Scal_S'')\le \Qcal(\Scal_S')$. Then, by Lemma~\ref{singleperiod2}, with $\Scal_S''$ being the greedy solution, $\Qcal(\Scal_S')$ being the capacity constraint and $\Scal_S'$ being the optimal filling of small items in $\Ical_S(1)$, we conclude that 
	$$
	{\Pcal}({\Scal}_S'') \ge {\Pcal}(\Scal_S') - \frac{1}{2T}\epsilon P_0.
	$$
	Therefore, $p' = \tilde{\Pcal}(\Scal') = \tilde{\Pcal}(\Scal_L'\cup \Scal_S')$ = $\tilde{\Pcal}(\Scal_L') + \Delta\tilde{\Pcal}(\Scal_S', c_1-\Qcal(\Scal_L')) \le \tilde{\Pcal}(\Scal_L'')+\Delta\tilde{\Pcal}({\Scal}_S'',c_1-\Qcal(\Scal_L')) + \frac{1}{2T}\epsilon P_0 + \kappa\le \tilde{\Pcal}(\Scal_L'')+\Delta\tilde{\Pcal}({\Scal}_S'',c_1-\Qcal(\Scal_L'')) + \frac{1}{2T}\epsilon P_0 + \kappa=\tilde{\Pcal}(\Scal_L''\cup {\Scal}_S'') + \frac{1}{2T}\epsilon P_0 + \kappa$.
	Let $p=\tilde{\Pcal}(\Scal_L''\cup {\Scal}_S'')$. From Algorithm~\ref{alg:FPTAS_nTlogn_small2}, we know that since $\Scal_S''$ includes the small items in $\widetilde{\mathcal{I}}_S(1)$ with the highest reward densities, the solution $\Scal_L''\cup\Scal_S''$ is one feasible solution for $\widetilde{A}_1(p)$. We thus have that 
	$$
	\widetilde{A}_1(p)\ge c_1-\Qcal(\Scal_L''\cup \Scal_S'')\ge c_1-\Qcal(\Scal')=\widetilde{C}_1(p'),
	$$
	where $p\ge p'-\frac{1}{2T}\epsilon R_0-\kappa$, and the second inequality follows from the facts that $\Qcal(\Scal_L'')\le \Qcal(\Scal_L')$ and $\Qcal(\Scal_S'')\le \Qcal(\Scal_S')$.
	
	For the induction step, assume that for all $p''\in \left\{ 0, 1, \ldots,\left\lceil\frac{16T}{\epsilon^2}\right\rceil \right\} \cdot \kappa$, we have that $\widetilde{A}_{t-1}(p) \ge \widetilde{C}_{t-1}(p'')$ for some $p\ge p''-\frac{1}{2T}\epsilon P_0(t-1)-\kappa (t-1)$. We want to show that for all $p'$, $\widetilde{A}_{t}(p) \ge \widetilde{C}_{t}(p')$ for some $p\ge p'-\frac{1}{2T}\epsilon P_0t-\kappa t$. Let $\Scal'$ be the solution corresponding to $\widetilde{C}_t(p')$, i.e., $$\Scal' :=\arg\max_{\left\{\substack{
			\Scal \subseteq \bigcup_{t'=1}^t\mathcal{I}(t')\ :\  \tilde{\Pcal}(\Scal)\ge p'}
		\right\}} \max_{0 \leq t' < t}\left\{ c_t - c_{t'} - \sum_{t'+1 \leq \tau \leq  t} \Qcal(\Scal(\tau)) \right\},$$ and let $\Scal_L'=\Scal'\cap \Ical_L$, $\Scal'_S= \Scal'\cap\Ical_S$.
	Let $\Scal'(t) := \{i\in \Scal'\mid d_i=t\}$ and consider the partial solution $\cup_{t'=1}^{t-1}\Scal'(t')$. By induction assumption, there exists some partial solution $\cup_{t'=1}^{t-1}\Scal''(t')$ such that $\Qcal\left(\cup_{t'=1}^{t-1}\Scal''(t')\right)\le \Qcal\left(\cup_{t'=1}^{t-1}\Scal'(t')\right)$, and that $\tilde{\Pcal}\left(\cup_{t'=1}^{t-1}\Scal''(t')\right)\ge \tilde{\Pcal}\left(\cup_{t'=1}^{t-1}\Scal'(t')\right)-\frac{1}{2T}\epsilon P_0(t-1)-\kappa(t-1)$. 
	
	First, we fill the partial solution $\cup_{t'=1}^{t-1}\Scal''(t')$ using items from $\Ical_L(t)$ according to Algorithm~\ref{alg:FPTAS_nTlogn_large2}. Note that one feasible solution is $\Scal_L'(t)$ which results in $\cup_{t'=1}^{t-1}\Scal''(t')\cup \Scal_L'(t)$. This keeps $\Qcal\left(\cup_{t'=1}^{t-1}\Scal''(t')\cup \Scal_L'(t)\right)\le \Qcal\left(\cup_{t'=1}^{t-1}\Scal'(t')\cup \Scal_L'(t)\right)$ 
	while having $\tilde{\Pcal}\left(\cup_{t'=1}^{t-1}\Scal''(t')\cup \Scal_L'(t)\right)\ge \tilde{\Pcal}\left(\cup_{t'=1}^{t-1}\Scal'(t')\cup \Scal_L'(t)\right)-\frac{1}{2T}\epsilon P_0(t-1)-\kappa(t-1)$. Suppose that after filling items from $\mathcal{I}_L(t)$ using DP in Algorithm~\ref{alg:FPTAS_nTlogn_large2}, the resulting set corresponding to $\hat{A}_t\left( \tilde{\Pcal}\left(\cup_{t'=1}^{t-1}\Scal''(t')\cup \Scal_L'(t)\right) \right)$ is $\widetilde{\Scal}$, then this $\widetilde{\Scal}$ would only use less space and earn more profit, i.e.,
	\begin{align*}
	\Qcal\left(\widetilde{\Scal}\right)&\le \Qcal\left(\cup_{t'=1}^{t-1}\Scal''(t')\cup \Scal_L'(t)\right)\le \Qcal\left(\cup_{t'=1}^{t-1}\Scal'(t')\cup \Scal_L'(t)\right),\\
	\tilde{\Pcal}\left(\widetilde{\Scal}\right)&\ge \tilde{\Pcal}\left(\cup_{t'=1}^{t-1}\Scal''(t')\cup \Scal_L'(t)\right)\\
	&\ge \tilde{\Pcal}\left(\cup_{t'=1}^{t-1}\Scal'(t')\cup \Scal_L'(t)\right)-\frac{1}{2T}\epsilon P_0(t-1)-\kappa(t-1).
	\end{align*} 
	
	Next, consider filling the partial solution $\widetilde{\Scal}$ using items from $\Ical_S(t)$. Let $\Scal_S''(t)$ be the small item set obtained by adding small items greedily in their reward densities, subject to the constraint that $\Qcal\left(\Scal_S''(t)\right)\le \Qcal\left(\Scal_S'(t)\right)$. Then, by Lemma~\ref{singleperiod2}, with $\Scal_S''(t)$ being the greedy solution, $\Qcal(\Scal_S'(t))$ being the capacity constraint and $\Scal_S'(t)$ being the optimal filling of small items in $\Ical_S(t)$, we conclude that 
	$$
	{\Pcal}({\Scal}_S''(t)) \ge {\Pcal}(\Scal_S'(t)) - \frac{1}{2T}\epsilon P_0.
	$$
	Therefore, 
	\begin{align*}
	p' &= \tilde{\Pcal}(\Scal') = \tilde{\Pcal}\left(\cup_{t'=1}^{t-1}\Scal'(t')\cup \Scal_L'(t)\cup \Scal_S'(t)\right) \\
	&\le \tilde{\Pcal}\left(\widetilde{\Scal}\cup\Scal_S''(t)\right)+\frac{1}{2T}\epsilon P_0(t-1)+\kappa(t-1) +  \frac{1}{2T}\epsilon P_0 + \kappa\\
	&\le \tilde{\Pcal}\left(\widetilde{\Scal}\cup {\Scal}_S''(t)\right) + \frac{1}{2T}\epsilon P_0t + \kappa t.
	\end{align*}
	Let $p=\tilde{\Pcal}\left(\widetilde{\Scal}\cup {\Scal}_S''(t)\right)$. From Algorithm~\ref{alg:FPTAS_nTlogn_small2}, we know that since $\Scal_S''(t)$ includes the small items in $\widetilde{\mathcal{I}}_S(t)$ with the highest reward densities, the solution $\widetilde{\Scal}\cup\Scal_S''(t)$ is one feasible solution for $\widetilde{A}_t(p)$. We thus have that 
	\begin{align*}
	\widetilde{A}_t(p)&\ge \max_{0 \leq t' < t}\left\{ c_t - c_{t'} - \sum_{t'+1 \leq \tau \leq  t} \Qcal\left(\left(\widetilde{\Scal}\cup {\Scal}_S''(t)\right)(\tau)\right) \right\}\\
	&\ge  \max_{0 \leq t' < t}\left\{ c_t - c_{t'} - \sum_{t'+1 \leq \tau \leq  t} \Qcal(\Scal'(\tau)) \right\}=\widetilde{C}_t(p'),
	\end{align*}
	where $p\ge p'-\frac{1}{2T}\epsilon P_0t-\kappa t$. This finishes the induction step, and thus the proof of the lemma.
\end{proof}

Using the above lemmas, we prove the following approximation result.
\begin{proposition}\label{MPBKPapprox2}
Let $\Scal'$ denote the optimal solution set by Algorithm~\ref{alg:FPTAS_nTlogn2}, i.e., $\Scal'$ is the solution set corresponding to $\widetilde{A}_T(p^*)$ where $p^*$ is the maximum $p$ such that $\widetilde{A}_T(p)>-\infty$. Let $\Scal^*$ be the optimal solution set to the original MPBKP-S. Then, 
\begin{align*}
\Pcal(\Scal')\ge p^*\ge (1-\epsilon-3\epsilon^2/8)\Pcal(\Scal^*).
\end{align*}
\end{proposition}
\begin{proof}
Note that $\tilde{\Pcal}(\Scal') = p^*$. Lemma~\ref{tilde2} implies that $$\widetilde{A}_T(p^*)\ge \widetilde{C}_T\left(p^*+\frac{1}{2T}\epsilon P_0T+\kappa T\right)=\widetilde{C}_T\left(p^*+\frac{1}{2}\epsilon P_0+\kappa T\right).$$ Since $\widetilde{C}_T(\tilde{\Pcal}(\Scal^*))>-\infty$, we have that $\widetilde{A}_T\left(\tilde{\Pcal}(\Scal^*)-\frac{1}{2}\epsilon P_0-\kappa T\right)\ge \widetilde{C}_T(\tilde{\Pcal}(\Scal^*))>-\infty$. Therefore,
$$
\Pcal(\Scal') \ge p^* \ge \tilde{\Pcal}(\Scal^*)-\frac{1}{2}\epsilon P_0 - \kappa T.
$$
By the definition of $\tilde{\Pcal}$ as in~\eqref{Ptilde}, for each large item, the reward is rounded down by at most $\kappa$ and the penalty is rounded up by at most $\kappa$, and all small items are together rounded down by at most $\kappa T$. Note that each large items earns profit $p_i$ unless it is paying more penalty than it would be by itself, which happens at most once at each period. Thus, there are at most~$\frac{2P_0}{\frac{1}{2T}\epsilon P_0} + T=\frac{4T}{\epsilon} +T$ number of large items, and thus the total number of rounding downs (for both large and small items) is bounded by $\frac{4T}{\epsilon}+2T$. Therefore, we have that $\Pcal(\Scal^*)\le \tilde{\Pcal}(\Scal^*)+\left(\frac{4T}{\epsilon}+2T\right)\kappa$. In conclusion, 
\begin{align*}
\Pcal(\Scal') \ge p^* &\ge \tilde{\Pcal}(\Scal^*)-\frac{1}{2}\epsilon P_0 -\kappa T\\&\ge \Pcal(\Scal^*)-\left(\frac{4T}{\epsilon}+2T\right)\kappa-\frac{1}{2}\epsilon P_0-T\kappa= \Pcal(\Scal^*)-\epsilon P_0-3T\kappa\\
&\ge \left(1-\epsilon-3\epsilon^2/8\right)\Pcal(\Scal^*).
\end{align*}
\end{proof}

It remains to validate Algorithm~\ref{alg:SC_FPTAS} in the search of $P_0$ which satisfies~\eqref{P0}. When Algorithm~\ref{alg:SC_FPTAS} terminates, it returns the last $p^*$ and the solution set $\mathcal{S}'$ corresponding to $\widetilde{A}_{T}(p^*)$. We then have the following lemmas. %The proofs are shown in Appendix~\ref{Appb}.
\begin{lemma}\label{while4}
Algorithm~\ref{alg:SC_FPTAS} terminates within $\log n$ iterations of the ``while'' loop (line~3).
\end{lemma}
\begin{proof}[Proof of Lemma~\ref{while4}]
	When $P_0$ satisfies~\eqref{P0}, by Proposition~\ref{MPBKPapprox2} we have that
	$$
	p^*\ge (1-\epsilon) \Pcal(\Scal^*)\ge (1-\epsilon)P_0.
	$$
	Thus, the ``while" loop terminates when $P_0$ satisfies~\eqref{P0}, if not before $P_0$ satisfies~\eqref{P0}. When $P_0$ satisfies~\eqref{P0}, we would also have $\Pcal(\Scal^*)/2\le P_0\le \Pcal(\Scal^*)$. Therefore, the number of iterations is upper bounded by
	$$
	\text{number of iterations}\le \log\frac{\bar{P}/2}{\Pcal(\Scal^*)/2}\le \log n,
	$$
	where we have used the fact that $\bar{P}\le nP\le n\Pcal(\Scal^*)$.
\end{proof}

\begin{lemma}\label{MPBKPSCfinallemma}
After running Algorithm~\ref{alg:SC_FPTAS}, suppose $\mathcal{S}'$ is the solution set corresponding to $\widetilde{A}_{T}(p^*)$, and $\mathcal{S}^*$ is the optimal solution set to the original MPBKP-S.  Then,
$$
\Pcal(\mathcal{S}')\geq  (1-\epsilon)\Pcal(\mathcal{S}^*).
$$
\end{lemma}
\begin{proof}[Proof of Lemma~\ref{MPBKPSCfinallemma}]
	If the ``while" loop terminates when $P_0>\Pcal(\Scal^*)$, i.e., it stops before $P_0$ falls below $\Pcal(\Scal^*)$, then we have that
	$$
	\Pcal(\Scal')\ge p^*\ge (1-\epsilon)P_0>(1-\epsilon)\Pcal(\Scal^*).
	$$
	Otherwise, from the proof of Lemma~\ref{while4} we know that the ``while" loop must terminate when $P_0$ first falls below $\Pcal(\Scal^*)$, which implies that the last $P_0$ satisfies~\eqref{P0}. Then by Proposition~\ref{MPBKPapprox2} we again have that 
	$$
	\Pcal(\Scal')\ge (1-\epsilon)\Pcal(\Scal^*).
	$$
	In either case, the solution we obtained from Algorithm~\ref{alg:SC_FPTAS} achieves $(1-\epsilon)$ optimal.
\end{proof}

With the above Lemmas, we are in a position to prove Theorem~\ref{main:MPBKP-S}.
\begin{proof}[Proof of Theorem~\ref{main:MPBKP-S}]
By Lemma~\ref{MPBKPSCfinallemma}, the solution found is within $(1-\epsilon)$ factor of $\Pcal(\Scal^*)$. Since the running time of the algorithm is $\mathcal{O}\left(n\cdot \left\lceil\frac{16T}{\epsilon^2}\right\rceil\cdot \log n\right)=\mathcal{O}\left(\frac{Tn\log n}{\epsilon^2}\right)$, which is polynomial in $n$ and $1/\epsilon$, the theorem follows. 
\end{proof}

\subsection{Proof of Theorem~\ref{thm:GRprofit}}\label{appc-unit}

This subsection is devoted to the proof of Theorem~\ref{thm:GRprofit}.
The idea is to look at the greedy solution set $\Scal_p$ and the optimal solution set $\Scal^*$, and by swapping each item in $\Scal_p$ to $\Scal^*$ in replacement of the same item or two other items, we construct a sequence of partial solutions of the greedy algorithm as well as modified optimal solution set, while maintaining the invariant that the profit of $\Scal^*$ is bounded by the sum of two times the profit of items in $\Scal_p$ swapped into $\Scal^*$ so far and the additional profit of remaining items in the modified optimal solution set. We will make this clear in the following. %Without loss of generality, in this proof we assume that $q_i = q=1$

To proceed, we first introduce some notations. Let $\Scal_p = \{g_1,\ldots,g_l\}$ and $\Scal^*=\{o_1,\ldots,o_m\}$, i.e., the items in greedy solution is denoted by $g_i$'s and the items in the optimal solution is denoted by $o_i$'s. Further, for any two sets of items $\Scal_1$ and $\Scal_2$, we define the incremental profit of adding $\Scal_2$ to the set $\Scal_1$ as
\begin{align}
\Delta \Pcal(\Scal_1,\Scal_2) = \Pcal(\Scal_1\cup\Scal_2) - \Pcal(\Scal_1).
\end{align}
Recall that $\Phi(\Scal)$ is the expected number of units of overflows that penalties are paid, which will be referred as \emph{overflow units} in the following. The incremental expected overflow units of adding $\Scal_2$ to the set $\Scal_1$ is defined as
\begin{align}
\Delta\Phi(\Scal_1,\Scal_2) = \Phi(\Scal_1\cup\Scal_2) - \Phi(\Scal_1).
\end{align}
%For a sample path of incremental capacities $\{a_t\}_{t=1}^T$, 
On a sample path of incremental capacities $\omega = \{c_t\}_{t=1}^T$, let $a_t := c_t - c_{t-1}$. Let $\Pcal_\omega$ and $\Phi_\omega$ be the profit and overflow units function, respectively, and the incremental profit of adding $\Scal_2$ to the set $\Scal_1$ is $$\Delta\Pcal_\omega(\Scal_1,\Scal_2) = \Pcal_\omega(\Scal_1\cup\Scal_2) - \Pcal_\omega(\Scal_1).$$
Similarly, on sample path $\omega$, the incremental penalty of adding $\Scal_2$ to the set $\Scal_1$ is $$\Delta\Phi_\omega(\Scal_1,\Scal_2) = \Phi_\omega(\Scal_1\cup\Scal_2) - \Phi_\omega(\Scal_1).$$
Then, the relationship of $\Delta\Pcal$ and $\Delta\Phi$ is:
\begin{align*}
\Delta\Pcal(\Scal_1,\Scal_2) &= \Pcal(\Scal_1\cup\Scal_2) - \Pcal(\Scal_1) = \Rcal(\Scal_1\cup\Scal_2) - \Rcal(\Scal_1) - B\cdot\Phi(\Scal_1\cup\Scal_2) + B\cdot\Phi(\Scal_1)\\
&= \Rcal(\Scal_2) - B\cdot \Delta\Phi(\Scal_1,\Scal_2).
\end{align*}
Similarly, on a sample path, we have that $\Delta\Pcal_\omega(\Scal_1,\Scal_2) = \Rcal(\Scal_2) - B\cdot \Delta\Phi_\omega(\Scal_1,\Scal_2)$.

Let $\Scal(t):=\{j\in\Scal\mid d_j=t\}$. Given a (partial) solution $\Scal$ and a sample path of capacities $\omega = \{c_t\}_{t=1}^T \in \Omega$. We let $a_t:= c_t-c_{t-1}$, and the available leftover capacity at time $t$ (after including items in $\Scal(t)$) is 
$$
\max\left\{\sup_{t' \le t} \sum_{\tau = t'}^t  a_\tau - \Qcal(S(\tau)), 0\right\} := \Ccal^\Scal_\omega(t).
$$ 
Then,  overflow units at time $t$ is 
$$
\max\left\{\sup_{t' \le t}\Qcal(S(\tau))- \sum_{\tau = t'}^t  a_\tau, 0\right\} := \Phi_\omega^\Scal(t),
$$
and the total overflow units is $\Phi_\omega(\Scal) = \sum_{t=1}^T\Phi_\omega^\Scal(t)$.

With the above definitions,
we first consider the calculation of overflows on a set $\Scal$ of items for a given sample path $\omega$. This is done in Algorithm~\ref{alg:unit-penalty}.

\begin{algorithm}[h]
	\footnotesize
	\caption{\sc{Overflow Assignment}}
	\label{alg:unit-penalty}
	\algsetblock[Name]{Parameters}{}{0}{}
	\algsetblock[Name]{Initialize}{}{0}{}
	\algsetblock[Name]{Define}{}{0}{}
	\begin{algorithmic}[1]
		\Parameters: Sample path of capacities $(c_1, \ldots, c_T) \in \mathbb{N}^T$, an arbitrary ordered list of requests $\Lcal = ( d_1, d_2, \ldots , d_n )$  
%		\Initialize: Remaining capacity $\mathbf{c}^r = (c_1^r, \ldots, c_T^r) \gets (a_1, \ldots, a_T)$ \Comment{$a_t=c_t-c_{t-1}$}
		\Initialize: Remaining capacity $\mathbf{a}^r = (a_1^r, \ldots, a_T^r) \gets (a_1, \ldots, a_T)$ \Comment{$a_t=c_t-c_{t-1}$}
		\Initialize: Units of overflow needing to pay penalty $\Phi \gets 0$
		\State $i \gets 1$
		\While {$i \leq n$}
		\State $q^r \gets q_i$
		\State $t_i = \max \{ t \leq d_i : a_t^r > 0 \}$  
		\While {$q^r > 0$}
		\If {$t_i < \infty$ and $t_i > 0$ }
		\State $a_{t_i}^r \gets a_{t_i}^r - \min\left\{a^r_{t_i}, q^r\right\}$
		\State $q^r \gets q^r - \min\left\{a^r_{t_i}, q^r\right\}$
		\State $t_i\gets t_i-1$
		\Else 
		\State $\Phi \gets \Phi+q^r$
		\State $q^r \gets 0$
		\EndIf
		\EndWhile
		\State $i \gets i + 1$
		\EndWhile
		\State {\bf Return} $(\mathbf{a}^r, \Phi)$
	\end{algorithmic}
\end{algorithm}

Algorithm~\ref{alg:unit-penalty} serves dual purpose -- while calculating the overflow, it also implicitly finds an assignment of the items which do not suffer a penalty to supply units. The assignment of items to supply units can be non-unique, while Algorithm~\ref{alg:unit-penalty} identifies one way of matching. Intuitively, the algorithm assigns items to the latest available units, saving the earlier capacity for items with shorter deadlines. This allows us to find the total overflows by considering the items in an arbitrary order (instead of in increasing order of deadlines), which is in turn useful for finding incremental profit~$\Delta\Pcal$ when we add a set of requests to an existing set of accepted requests.
%In particular, if $t_i \leq T$, then item $i$ is assigned to a capacity unit from time $t_i$. 
We begin with the following lemma which proves that Algorithm~\ref{alg:unit-penalty} indeed finds the minimum overflow.

\begin{lemma} 
	\label{lem:unit_penalty}
	Given a sample path $\omega \in \NN^{T}$ of supply, and a set $\Scal $ of items with general integer demands, let $\Lcal = (d_1, \ldots, d_n)$ be an arbitrary ordering of the items in $\Scal$ ($d_i$ denoting the deadlines). Then the overflow units $\Phi$ returned when executing Algorithm~\ref{alg:unit-penalty} ({\sc Overflow Assignment}) on $(\omega, \Lcal)$ satisfies $\Phi = \Phi_{\omega}({\Scal})$.
\end{lemma}
\begin{proof}[Proof of Lemma~\ref{lem:unit_penalty}]
    We will use LP duality to prove the Lemma. In a nutshell, we will use the the assignment created by Algorithm~\ref{alg:unit-penalty} to create a feasible solution to the dual LP such that the objective function of the dual matches the objective function  penalty of the assignment. Since any feasible solution of the dual lower bounds the optimal, we would have thus demonstrated the optimality of the assignment and hence of the overflow units $\Phi$. 

\begin{align*}
\begin{array}{rl}
& \mbox{(PRIMAL)} \\
\min & \sum_{i=1}^n y_i \\
\mbox{s.t.} &  \\
\forall t \in [T] : & -\sum_{i: d_i \leq t} x_i  \geq - c_t\\
\forall i \in [n]: & x_i + y_i = q_i \\
& x_i, y_i \geq 0
\end{array}
\hspace{0.2in}
\left|
\hspace{0.2in}
\begin{array}{rl}
& \mbox{(DUAL)} \\
\max & \sum_{i=1}^n q_i \gamma_i - \sum_{t} \lambda_t c_t  \\
\mbox{s.t.} & \\
\forall i \in [n]: & \gamma_i \leq 1 \\
\forall i \in [n]: & \gamma_i \leq  \sum_{t \geq d_i} \lambda_t \\
& \lambda_t \geq 0
\end{array}
\right.
\end{align*}

To construct the dual solution, let $\tau = \min \{ t : a_t^r > 0\}$. That is, $\tau$ is the first time at which there is some capacity remaining after the assignment of {\sc Overflow Assignment}. By the nature of the algorithm, there are no items with $d_i \geq \tau$ for which penalty is paid, and in fact all items with $d_i \geq \tau$ are served with capacity that arrives at time $\tau$ or later. Therefore, the overflow units under the assignment is the total size of items with $d_i < \tau$ minus the capacity $c_{\tau-1}$ (since this capacity is only used by requests with $d_i < \tau$).

Now construct a dual solution as follows:
\begin{align*}
\lambda_t = \begin{cases}
1 & t = \tau-1, \\
0 & t \neq \tau-1;
\end{cases} 
\qquad 
\gamma_i = \begin{cases}
1 & d_{i} \leq \tau-1 , \\
0 & d_{i} \geq \tau.
\end{cases}
\end{align*}
It is easy to verify that this is a feasible dual solution. Further, the objective function value under this feasible dual is
\[ \sum_{ i : d_{i} \leq \tau -1} q_i  - c_{\tau-1} \]
which is exactly the overflow units of the primal assignment. Therefore, the primal solution in fact attains the optimal objective. 
\end{proof}

As a result of Algorithm~\ref{alg:unit-penalty} and Lemma~\ref{lem:unit_penalty}, we have the following lemma.
\begin{lemma}\label{lem:morecap}
	Let ${\Scal}$ be a set of items disjoint with $\Scal_1$ and $\Scal_2$. If for some $\omega=\{c_t\mid t\in[T]\}\in\Omega$, we have $\Ccal_\omega^{\Scal_1}(t)\ge \Ccal_\omega^{\Scal_2}(t), \forall t\in[T]$, then, $\Delta\Pcal_\omega(\Scal_1,{\Scal})\ge \Delta\Pcal_\omega(\Scal_2,{\Scal})$. If this is true for all $\omega\in \Omega$, we further have that $\Delta\Pcal(\Scal_1,{\Scal})\ge \Delta\Pcal(\Scal_2,{\Scal})$.		
	%Let $\bar{\Scal} = \{s_1,\ldots,s_k\}$, where $s_i$ has deadline $d_i$. For each realization of incremental capacities $\{a_t\}_{t=1}^T$, if the remaining capacities at $t=1,\ldots, T$ are larger after serving $\Scal_1$ than after serving $\Scal_2$, then the incremental profit of adding $\bar{\Scal}$ to $\Scal_1$ is no smaller than adding $\bar{\Scal}$ to $\Scal_2$. In other words, let $\Scal(t):=\{j\in\Scal\mid d_j=t\}$. If $\left[c_{t} - |{\Scal}_1(1)\cup\cdots\cup\Scal_1(t)|\right]^+\ge \left[c_{t} - |{\Scal}_2(1)\cup\cdots\cup\Scal_2(t)|\right]^+$ for all $t\in[T]$, then $\Delta\Pcal(\Scal_1,\bar{\Scal})\ge \Delta\Pcal(\Scal_2,\bar{\Scal})$.
\end{lemma}
\begin{proof}[Proof of Lemma~\ref{lem:morecap}]
It suffices to show that $\Delta\Phi_\omega(\Scal_1,\Scal) \le \Delta\Phi(\Scal_2,\Scal)$. Note that 
\begin{align*}
    \Delta\Phi_\omega(\Scal_1,\Scal) &= \Phi_{\omega'}(\Scal), \text{ where }\omega' =\left\{\Ccal^{\Scal_1}_\omega(t)\mid t\in[T]\right\},\\
    \Delta\Phi_\omega(\Scal_2,\Scal) &= \Phi_{\omega''}(\Scal), \text{ where }\omega'' =\left\{\Ccal^{\Scal_2}_\omega(t)\mid t\in[T]\right\}.
\end{align*}
By Lemma~\ref{lem:unit_penalty}, the ordering of items in $\Scal$ does not matter when computing the total overflow units, and we may apply Algorithm~\ref{alg:unit-penalty} to compute $\Phi_{\omega'}(\Scal)$ and $\Phi_{\omega''}(\Scal)$. Since $\Ccal_\omega^{\Scal_1}(t)\ge \Ccal_\omega^{\Scal_2}(t), \forall t\in[T]$, as we apply Algorithm~\ref{alg:unit-penalty}, for any capacity in $\omega''$ that is used to serve a unit of demand in $\Scal$, we have the same capacity in $\omega'$ that can be used to serve the same unit of demand in $\Scal$. It then follows that $\Phi_{\omega'}(\Scal) \le \Phi_{\omega''}(\Scal)$, which implies that $\Delta\Phi_\omega(\Scal_1,\Scal)\le \Delta\Phi_\omega(\Scal_2,\Scal)$.
\end{proof}

\endproof

We next show the submodularity of $\Pcal$.
\begin{lemma}\label{lem:submod}
	For any $\Scal_1\subseteq\Scal_2$, we have that $\Delta\Pcal(\Scal_1,\Scal_3)\ge \Delta\Pcal(\Scal_2,\Scal_3)$.
\end{lemma}
\begin{proof}[Proof of Lemma~\ref{lem:submod}]
%	By definition, 
%	\begin{align*}
%		\Delta \Pcal(\Scal_1,\Scal_3) &= \Pcal(\Scal_1\cup\Scal_3) - \Pcal(\Scal_1) = \Rcal(\Scal_1\cup \Scal_3)-\Rcal(\Scal_1) - \Phi(\Scal_1\cup\Scal_3)+\Phi(\Scal_1)\\
%		 &= \Rcal(\Scal_3)- \Phi(\Scal_1\cup\Scal_3)+\Phi(\Scal_1).\\
%		\Delta \Pcal(\Scal_2,\Scal_3) &= \Pcal(\Scal_2\cup\Scal_3) - \Pcal(\Scal_2)= \Rcal(\Scal_2\cup \Scal_3)-\Rcal(\Scal_2) - \Phi(\Scal_2\cup\Scal_3)+\Phi(\Scal_2)\\
%		&= \Rcal(\Scal_3)- \Phi(\Scal_2\cup\Scal_3)+\Phi(\Scal_2).
%	\end{align*}
Since $\Scal_1\subseteq\Scal_2$, in each realized sample path of capacities $\omega = \{c_t\}_{t=1}^T$, it should be clear that $\Ccal_\omega^{\Scal_1}(t)\ge \Ccal_\omega^{\Scal_2}(t), \forall t$, i.e., at each time period, the available remaining capacity on $\Scal_1$ is no less than the available remaining capacity on $\Scal_2$. Thus, by Lemma~\ref{lem:morecap}, the result follows.
\end{proof}

Lemma~\ref{lem:morecap} and Lemma~\ref{lem:submod} showed the relationship of incremental profit change of adding a set of items on top of two other sets of items. Specifically, if one set always has more remaining capacity than the other set, then adding a third set to one generates more incremental profit than adding the same set to the other.

For the rest of this section, we impose the assumption that $q_i = q, \forall i\in[N]$. To simplify the presentation, we may without loss of generality assume that $q=1$ by allowing $\{c_t\}$ to be nonintegers.
We next have the following result which will serve as a key to prove Theorem~\ref{thm:GRprofit}.

\begin{lemma}\label{lem:incre_2}
	Let $\Scal_1$ and $\Scal_2 = {\Scal}_2^-\sqcup {\Scal}_2^+$ be two disjoint set of items. Let $i, j, k$ be three items not in either set such that:
	\begin{enumerate}
		\item $d_m \leq d_j \leq d_i$, for all items $m \in \Scal_2^-$,
		\item $d_i \leq d_k \leq d_m$, for all items $m \in \Scal_2^+$.
	\end{enumerate}
	Then, we have that
	\begin{align}
	\Delta\Pcal\left(\Scal_1\cup \{j,k\}, \Scal_2\right)\le \Delta\Pcal\left(\Scal_1\cup\{i\}, \Scal_2 \right).
	\end{align}
\end{lemma}
\begin{proof}[Proof of Lemma~\ref{lem:incre_2}]
    We begin with two observations.\\
{\bf Observation 1:} Using Lemma~\ref{lem:unit_penalty}, we can determine $\Delta \Pcal\left(  \Scal_1 \cup \{j,k\} , \Scal_2 \right)$ as follows: We first fix an ordering of $\Scal_1$ and assign them using Algorithm~\ref{alg:unit-penalty}. This gives some residual capacity vector $\mathbf{c}^r$. The problem of finding $\Delta \Pcal(\Scal_1 \cup \{j,k\}, \Scal_2)$ under capacity vector $\omega$ now reduces to finding $\Delta \Pcal(\{j,k\}, \Scal_2)$ under capacity vector $\mathbf{c}^r$. Similarly,  finding $\Delta \Pcal(\Scal_1 \cup \{i\}, \Scal_2)$ under capacity vector $\omega$ reduces to finding $\Delta \Pcal(\{i\}, \Scal_2)$ under capacity vector $\mathbf{c}^r$.\\ 

{\bf Observation 2:} It suffices to prove the Lemma for $|\Scal_{2}|=1$.

We therefore consider two cases, based on whether the item $m$ in $S_2$ has $d_m \leq d_j \leq d_i$ or $d_m \geq d_k \geq d_i$. Note that we have reduced to a case where we only need to worry about items $i,j,k,m$ and capacity availability $\mathbf{c}^r$. 

{\bf Case : $d_m \leq d_j \leq d_i$}\\
To find incremental penalty:
\[ \Phi^{\{i,m\}}_{\mathbf{c}^r} - \Phi^{\{i\}}_{\mathbf{c}^r} \]
we will first add item $i$ and then $m$ according to Algorithm~\ref{alg:unit-penalty}. Similarly, for 
\[ \Phi^{\{j,k,m\}}_{\mathbf{c}^r} - \Phi^{\{j,k\}}_{\mathbf{c}^r} \]
we first add item $j$, then $k$ and then $m$. 
We claim that if item $m$ does not pay a penalty in the latter case (when added to $\{j,k\}$), then it does not pay a penalty when added to $\{i\}$. To see why, if $m$ does not pay a penalty when added to $\{j,k\}$, then it must be that 
\[  \sum_{t \leq d_j} c_t^r \geq 2, \quad \sum_{t \leq d_m} c_t^r \geq 1. \]
In this case, when adding item $i$, there is still residual capacity left for matching $m$.

{\bf Case : $d_i \leq d_k \leq d_m$ }\\ 
In this we argue that if $m$ pays a penalty when added to $i$, then it must pay a penalty when added to $\{j,k \}$. If $m$ pays penalty for $i$, then:
\[  \sum_{t \leq d_i} c_t^r \leq 1 , \quad \sum_{ d_i < t \leq d_m } c_t^r = 0. \]
In this case when we first add $k$, it uses up any capacity $c_t^r \leq d_i$, leaving $m$ to pay a penalty.

Therefore, in either case, the incremental overflow units when adding item $m$ to item $i$ is at most the incremental overflow units when adding $m$ to $\{j,k\}$. 
\end{proof}

With the above lemmas, we are in a position to prove Theorem~\ref{thm:GRprofit}. 
\begin{proof}[Proof of Theorem~\ref{thm:GRprofit}]
    First, suppose that without loss of generality, the items in $\Scal_p$ are added exactly in the order of $g_1,\ldots,g_l$. Our proof is done by defining $G_i$ and $\Scal^*_i$ inductively, and show that 
\begin{align*}
\Pcal(\Scal^*)\le 2\Pcal(G_i) + \Delta\Pcal(G_i,\Scal^*_i),\quad \forall i\le \min\{l, m\} \text{ s.t. }\Scal_i^* \text{ is well-defined}.
\end{align*}
\emph{Base Case.}
Let $G_1 = \{g_1\}$ and let $\Scal^* = {\Scal^*}^-\sqcup {\Scal^*}^+$ where ${\Scal^*}^-:=\{j\in \Scal^*\mid d_j<d_{g_1}\}$ and ${\Scal^*}^+:=\{j\in \Scal^*\mid d_j\ge d_{g_1}\}$. Define
\begin{align*}
\Scal^*_1 = \begin{cases}
\Scal^*\setminus \{g_1\},\quad \text{if }g_1\in \Scal^*\\
\Scal^*\setminus \{o',o''\},\quad \text{if }g_1\notin \Scal^*
\end{cases}
\end{align*}
where $o'\in\Scal^*: d_{i'}\le d_{o'}\le d_{g_1},\forall i'\in {\Scal^*}^-$, and $o''\in\Scal^*: d_{g_1}\le d_{o'}\le d_{j'}, \forall j'\in{\Scal^*}^+, $ i.e., $o'$ is an item in $\Scal^*$ with deadline no later than $g_1$ but no earlier than the deadlines of items in ${\Scal^*}^-$, and $o''$ is an item in $\Scal^*$ with deadline no earlier than $g_{1}$ but no later than the deadlines of items in ${\Scal^*}^+_{i}$ (if such $o'$ or $o''$ does not exist, then simply ignore it). Then, we have the two cases:
\begin{itemize}
	\item $g_1\in\Scal^*$.
	\begin{align*}
	\Pcal(\Scal^*) = \Delta\Pcal\left(\emptyset, \Scal^*\right) &= \Delta\Pcal\left(\emptyset, \{g_1\}\right) + \Delta\Pcal\left(\{g_1\},\Scal_1^*\right)\\
	&\le 2\Pcal(G_1) + \Delta\Pcal(G_1,\Scal^*_1)
	\end{align*}
	where the inequality follows directly from the fact that $\Pcal(G_1) = \Delta\Pcal\left(\emptyset, \{g_1\}\right)$ is nonnegative.
	\item $g_1\notin \Scal^*$. First note that 
	\begin{align*}
	\Delta\Pcal\left(\emptyset, \left\{o',o''\right\}\right)&= \Delta\Pcal\left(\emptyset, \{o'\}\right) + \Delta\Pcal\left(\{o'\},\{o''\}\right)\\
	&\le \Delta\Pcal\left(\emptyset, \{o'\}\right) + \Delta\Pcal\left(\emptyset,\{o''\}\right)\\
	&\le \Delta\Pcal(\emptyset,\{g_1\}) + \Delta\Pcal(\emptyset,\{g_1\}) = 2\Delta\Pcal(\emptyset, \{g_1\}) = 2\Pcal(G_1),
	\end{align*}
	where the first inequality follows from Lemma~\ref{lem:submod} and the second inequality follows from the greedy algorithm that $g_1$ gives the greatest incremental profit.
	
	On the other hand, by Lemma~\ref{lem:incre_2}, we also have that 
	\begin{align*}
	\Delta\Pcal\left(\left\{o',o''\right\},\Scal_1^*\right)\le \Delta\Pcal\left(G_1,\Scal^*_1\right).
	\end{align*}
	
	Combining the above two inequalities, we conclude that
	\begin{align*}
	\Pcal(\Scal^*) = \Delta\Pcal\left(\emptyset, \Scal^*\right)&= \Delta\Pcal\left(\emptyset, \left\{o',o''\right\}\right) + \Delta\Pcal\left(\left\{o',o''\right\}, \Scal^*_1\right) \\
	&\le 2\Pcal\left(G_1\right) + \Delta\Pcal\left(G_1,\Scal^*_1\right)
	\end{align*}	
\end{itemize}

\emph{Induction Step.} Assume that $\Pcal(\Scal^*)\le 2\Pcal(G_i) + \Delta\Pcal(G_i,\Scal^*_i)$, we define $G_{i+1} = G_i \cup \{g_{i+1}\}$ and let $\Scal^*_i = {\Scal^*_i}^-\sqcup {\Scal^*_i}^+$ where ${\Scal^*_i}^-:=\{j\in \Scal^*_i\mid d_j<d_{g_i}\}$ and ${\Scal^*_i}^+:=\{j\in \Scal^*_i\mid d_j\ge d_{g_i}\}$. Define
\begin{align*}
\Scal^*_{i+1} = \begin{cases}
\Scal^*_i\setminus \{g_{i+1}\},\quad \text{if }g_{i+1}\in \Scal^*_i\\
\Scal^*_i\setminus \{o',o''\},\quad \text{if }g_{i+1}\notin \Scal^*_i
\end{cases}
\end{align*}
where where $o'\in\Scal^*_i: d_{i'}\le d_{o'}\le d_{g_i},\forall i'\in {\Scal^*_i}^-$, and $o''\in\Scal^*_i: d_{g_i}\le d_{o'}\le d_{j'}, \forall j'\in{\Scal^*_i}^+, $ i.e., $o'$ is an item in $\Scal^*_i$ with deadline no later than $g_i$ but no earlier than the deadlines of items in ${\Scal^*_i}^-$, and $o''$ is an item in $\Scal^*_i$ with deadline no earlier than $g_{i}$ but no later than the deadlines of items in ${\Scal^*}^+_{i}$ (if such $o'$ or $o''$ does not exist, then simply ignore it). Then, we have in the two cases:
\begin{itemize}
	\item $g_{i+1}\in\Scal^*_i$.
	\begin{align*}
	\Pcal(\Scal^*) &\le 2\Pcal(G_i) + \Delta\Pcal(G_i,\Scal^*_i)= 2\Pcal(G_i) + \Delta\Pcal\left(G_i, \left\{g_{i+1}\right\}\right) + \Delta\Pcal\left(G_{i+1},\Scal^*_{i+1}\right)\\
	&\le 2\Pcal(G_i) + 2\Delta\Pcal\left(G_i, \left\{g_{i+1}\right\}\right) + \Delta\Pcal\left(G_{i+1},\Scal^*_{i+1}\right)= 2\Pcal(G_{i+1}) + \Delta\Pcal\left(G_{i+1},\Scal^*_{i+1}\right)
	\end{align*}
	where the first inequality follows from the induction assumption and the second inequality follows directly from the fact that $\Delta\Pcal\left(G_i, \{g_{i+1}\}\right)$ is nonnegative.
	\item $g_{i+1}\notin \Scal^*_i$. First note that 
	\begin{align*}
	\Delta\Pcal\left(G_i, \left\{o',o''\right\}\right)&= \Delta\Pcal\left(G_i, \{o'\}\right) + \Delta\Pcal\left(G_i\cup\{o'\},\{o''\}\right)\\
	&\le \Delta\Pcal\left(G_i, \{o'\}\right) + \Delta\Pcal\left(G_i,\{o''\}\right)\\
	&\le \Delta\Pcal(G_i,\{g_{i+1}\}) + \Delta\Pcal(G_i,\{g_{i+1}\}) = 2\Delta\Pcal(G_i, \{g_{i+1}\}),
	\end{align*}
	where the first inequality follows from Lemma~\ref{lem:submod} and the second inequality follows from the greedy algorithm that $g_{i+1}$ adds the greatest incremental profit to $G_{i}$.
	
	On the other hand, by Lemma~\ref{lem:incre_2}, we also have that 
	\begin{align*}
	\Delta\Pcal\left(G_i\cup\left\{o',o''\right\},\Scal_{i+1}^*\right)\le \Delta\Pcal\left(G_{i+1},\Scal^*_{i+1}\right).
	\end{align*}
	
	Combining the above two inequalities, we conclude that
	\begin{align*}
	\Pcal(\Scal^*) &\le 2\Pcal(G_i) + \Delta\Pcal(G_i,\Scal^*_i)=  2\Pcal(G_i) + \Delta\Pcal\left(G_i,\left\{o',o''\right\}\right) + \Delta\Pcal\left(G_i\cup \left\{o',o''\right\},\Scal^*_{i+1}\right)\\
	&\le 2\Pcal(G_i) + 2\Delta\Pcal\left(G_i,\left\{g_{i+1}\right\}\right) + \Delta\Pcal\left(G_{i+1},\Scal^*_{i+1}\right)\\
	&\le 2\Pcal\left(G_{i+1}\right) + \Delta\Pcal\left(G_{i+1},\Scal^*_{i+1}\right)
	\end{align*}	
\end{itemize}
This completes the induction step. Note that at each step, $\Scal^*_{i+1}\subsetneq \Scal^*_i$ and $G_i\subsetneq G_{i+1}$. In the end, we will reach some $i'$ such that either $\Scal^*_{i'}=\emptyset$ or $G_{i'} = \Scal_p$ and $\Scal^*_{i'}\ne \emptyset$. In the first case, we have that
\begin{align*}
\Pcal(\Scal^*)&\le 2\Pcal(G_{i'}) + \Delta\Pcal(G_{i'},\Scal^*_{i'}) = 2\Pcal(G_{i'}) + 0\le 2\Pcal(\Scal_p).
\end{align*}
In the second case, i.e., $G_{i'} = \Scal_p$ and $\Scal^*_{i'}\ne \emptyset$, we again have that $\Pcal(\Scal^*)\le 2\Pcal(G_{i'}) + \Delta\Pcal(G_{i'},\Scal^*_{i'})$. Now if $ \Delta\Pcal(G_{i'},\Scal^*_{i'})>0$, then we can add the items in $\Scal^*_{i'}$ to $\Scal_p$ and still increase the profit, which violates the greedy algorithm. Thus, it must be that $\Delta\Pcal(G_{i'},\Scal^*_{i'})\le 0$. Then we would have
\begin{align*}
\Pcal(\Scal^*)&\le 2\Pcal(G_{i'}) + \Delta\Pcal(G_{i'},\Scal^*_{i'}) \le 2\Pcal(\Scal_p).
\end{align*}
In conclusion, we have that $\Pcal(\Scal^*)\le 2\Pcal(\Scal_p)$, or equivalently $\Pcal(\Scal_p)\ge \frac{1}{2}\Pcal(\Scal^*)$. This completes the proof of Theorem~\ref{thm:GRprofit}.    
\end{proof}

\section{Alternative FPTAS for MPBKP}\label{appT2}
In this section, we introduce another FPTAS for MPBKP, which has time complexity $\tilde{\Ocal}\left(n+\frac{T^2}{\epsilon^3}\right)$. To roughly describe the main idea, we will again adopt the functional approach to approximate~\eqref{eqn:ft}. Instead of having an approximation of $f_{\Ical(t)}$ for each $t$ directly from Lemma~\ref{lem:01}, we further partition $\Ical(t)$ into $m+1$ subsets ($m$ being specified later), i.e., $\Ical(t):=\Ical(t)_0\sqcup \Ical(t)_1\sqcup\cdots\sqcup \Ical(t)_m$, where items in each subset have approximately the same reward. Then, we have that $f_{\Ical(t)} = f_{\Ical(t)_0}\oplus f_{\Ical(t)_1}\oplus\cdots \oplus f_{\Ical(t)_m}:=\oplus_{j=0}^mf_{\Ical(t)_j}$, and by noting that the $(\max,+)$-convolution $\oplus$ is commutative, the function $f_t$ as defined in~\eqref{eqn:ft} can be computed as 
\begin{align}\label{eq:ftex}
f_t := \begin{cases}
f_{\Ical(1)}^{c_1} & t=1,\\
\left(f_{t-1}\oplus f_{\Ical(t)}\right)^{c_t} = \left(f_{t-1}\oplus f_{\Ical(t)_0}\oplus f_{\Ical(t)_1}\oplus\cdots \oplus f_{\Ical(t)_m}\right)^{c_t} & t\ge 2,
\end{cases}
\end{align}
and~\eqref{eq:ftex} can be computed more efficiently due to some special properties of $f_{\Ical(t)_j}$.

Before proceeding to the actual algorithm, we first have some preliminaries. A monotone step function $f_{\Ical}(c)$ with steps at $c_{1},c_{2},\ldots,c_l$ is called \emph{$r$-uniform} if it satisfies both of the following conditions:
\begin{enumerate}
	\item $\forall c\in\mathbb{R}^+$, $f_{\Ical}(c) = kr$ for some nonnegative integer $k$,
	\item $\exists c_{j} \text{ s.t. } f_{\Ical}(c_j)=kr \Longrightarrow \exists c_{j'} \text{ s.t. }f_{\Ical}(c_{j'}) = k'r, \forall k'\le k$ nonnegative integers.
\end{enumerate}
The monotone step function $f_{\Ical}(c)$ with steps at $c_{1},c_{2},\ldots,c_l$ is called \emph{pseudo-concave} if $c_{j+2} - c_{j+1}\ge c_{j+1}-c_{j}, \forall j=1,\ldots, l-2$. The \emph{range} of a function $f$ is the set of all possible function values. We then introduce the following lemma from~\cite{chan:OASIcs:2018:8299} for approximating $f\oplus g$ when $g$ is $r$-uniform and pseudo-concave. 

\begin{lemma}[\cite{chan:OASIcs:2018:8299}]\label{lem:02}
	Let $f$ and $g$ be monotone step functions with total complexity $l$ and ranges contained in $\{-\infty,0\}\cup\{A,B\}$. Then we can compute a monotone step function that approximates $f\oplus g$ with factor $1+\Ocal(\epsilon')$ and complexity $\tilde{\Ocal}\left(\frac{1}{\epsilon'}\right)$ in $\Ocal(l) + \tilde{\Ocal}\left(\frac{1}{\epsilon'}\right)$ time if $g$ is $r$-uniform and pseudo-concave.
\end{lemma}
With the above lemma, we present Algorithm~\ref{alg:MPBKP'} for MPBKP. 

\begin{algorithm}[ht]
	\footnotesize
	\caption{FPTAS for MPBKP in $\tilde{\Ocal}\left(n+T^2/\epsilon^3\right)$}
	\label{alg:MPBKP'}
	\algsetblock[Name]{Parameters}{}{0}{}
	\algsetblock[Name]{Initialize}{}{0}{}
	\algsetblock[Name]{Define}{}{0}{}
	\begin{algorithmic}[1]
		\Statex \textbf{Input:} $[n], c_1,\ldots, c_T$  \Comment Set of items to be packed, cumulative capacities up to each time $t$
		\Statex \textbf{Output:} $\tilde{f}_t$ \Comment Approximation of function $f_t$
		\State Discard all items with $r_i\le \frac{\epsilon}{n}\max_jr_j$ and relabel the items 
		\State $r_0\gets \min_ir_i$ \Comment Lower bound of solution value
		\State $\hat{r}_i\gets r_0\cdot (1+\epsilon)^{\left\lfloor\log_{1+\epsilon}\left(\frac{r_i}{r_0}\right)\right\rfloor}$
		\Comment Round down the reward of each item
		\State $m\gets \left\lceil\log_{1+\epsilon}\frac{n^2}{\epsilon}\right\rceil$ \Comment Number of distinct rewards to be considered, each in the form $r_0\cdot(1+\epsilon)^k$
		%		\State Initialize $\tilde{A}(0,r) = \begin{cases}
		%		0 & r = 0,\\
		%		-\infty & r > 0.
		%		\end{cases}$
		%\State Obtain $\tilde{f}_{\Ical(1)}$ that approximates $f_{\Ical(1)}$ with factor $1+\epsilon$ using Lemma~\ref{lem:01}
		%\State $\tilde{f}_1:= \tilde{f}_{\Ical(1)}^{c_1}$ \Comment $\tilde{f}_1$ has complexity at most $m=\tilde{\mathcal{O}}\left(\frac{1}{\epsilon}\right)$
		\State $\tilde{f}_0\gets -\infty$
		\For {$t=1,\ldots, T$}
		\State $\hat{f}_t\gets \tilde{f}_{t-1}$
		\For {$j=0,\ldots,m$}
		\State $\Ical(t)_j = \left\{i\in \Ical(t)\mid \hat{r}_i = r_0\cdot (1+\epsilon)^j\right\}$
		\Comment Items in each $\Ical(t)_j$ has the same rounded reward
		\State $\hat{\Ical}(t)_j = \left\{(\hat{r}_i,q_i)\mid i\in \Ical(t)_j\right\}$ and obtain $f_{\hat{\Ical}(t)_j}$
		\Comment Using items with rounded rewards, build the function $f_{\hat{\Ical}(t)_j}$
		\State Approximately compute $\hat{f}_t = \hat{f}_t\oplus f_{\hat{\Ical}(t)_j}$ using Lemma~\ref{lem:02}
		\EndFor
		\State $\tilde{f}_t = \hat{f}_t^{c_t}$
		\Comment $\tilde{f}_t$ is an approximation of $f_t$
		%		\State Obtain $\tilde{f}_{\Ical(t)}$ that approximates $f_{\Ical(t)}$ with factor $1+\epsilon$ using Lemma~\ref{lem:01}
		%		\State $l\gets$ complexity of $\tilde{f}_{\Ical(t)}$ \Comment $l=\tilde{\mathcal{O}}\left(\frac{1}{\epsilon}\right)$
		%		\State Compute $\hat{f}_{t}:= \left(\tilde{f}_{t-1}\oplus \tilde{f}_{\Ical(t)}\right)^{c_t}$, taking $m\cdot l$ time \Comment $\hat{f}_t$ has complexity $\tilde{\mathcal{O}}\left(\frac{1}{\epsilon^2}\right)$
		%		\State $\tilde{f}_t := r_0\cdot (1+\epsilon)^{\left\lfloor\log_{1+\epsilon}\left(\frac{\hat{f}_t}{r_0}\right)\right\rfloor}$ \Comment Round $\hat{f}_t$ down to the nearest $r_0\cdot (1+\epsilon)^k$ for $k=0,\ldots,m$. 
		%		\NoNumber{}\Comment Now $\tilde{f}_t$ has complexity at most $m=\tilde{\mathcal{O}}\left(\frac{1}{\epsilon}\right)$
		\EndFor
	\end{algorithmic}
\end{algorithm}

In Algorithm~\ref{alg:MPBKP'}, we first discard all items with reward $r_i\le \frac{\epsilon}{n}\max_jr_j$. The maximum we could lose is $n\cdot \frac{\epsilon}{n}\max_jr_j = \epsilon\max_jr_j$, which is at most $\epsilon$ fraction of the optimal value. We next round down the rewards of all remaining items  to the nearest $r_0\cdot (1+\epsilon)^k$, where $r_0:=\min_jr_j$ and $k$ is some nonnegative integer, so we lose at most a fraction of $(1+\epsilon)$ in the rounding, and the number of distinct rounded rewards is bounded by $m = \left\lceil\log_{1+\epsilon}\frac{n^2}{\epsilon}\right\rceil = \tilde{\Ocal}\left(\frac{1}{\epsilon}\right)$. We begin with initializing $\tilde{f}_0=-\infty$. Then, for period $t=1$, we partition $\Ical(1)=\sqcup_{j=0}^m \Ical(1)_j$ where all items in $\Ical(1)_j$ have rounded reward $r_0\cdot (1+\epsilon)^j$. Denote by $\hat{\Ical}(1)_j$ these items with rounded rewards, and by adding these items greedily in nonincreasing order of their sizes, we obtain $f_{\hat{\Ical}(1)_j}$, which is a $(1+\epsilon)$ approximation of $f_{\Ical(1)_j}$, and is $r_0\cdot (1+\epsilon)^j$-uniform and pseudo-concave. 
%Moreover, $f_{\hat{\Ical}(1)_j}$ is a $(1+\epsilon)$ approximation of $f_{\Ical(1)_j}$. 
By applying Lemma~\ref{lem:02} for $m+1$ times (with $\epsilon'$ to be specified later), we obtain $\tilde{f}_0\oplus f_{\hat{\Ical}(1)_0}\oplus f_{\hat{\Ical}(1)_1}\oplus\cdots \oplus f_{\Ical(1)_m} = \tilde{f}_0\oplus f_{\hat{\Ical}(1)}$, which approximates $f_0\oplus f_{\Ical(1)}$ with an accumulative approximation factor $(1+\epsilon)(1+\epsilon')^{m+1}$, and is computed in total time $\Ocal(n_1) + \tilde{\Ocal}\left(\frac{m+1}{\epsilon'}\right)$. Then, to ensure feasibility, $\tilde{f}_1$ is obtained by taking truncation $c_1$ on $\tilde{f}_0\oplus f_{\hat{\Ical}(1)}$, which becomes a $(1+\epsilon)(1+\epsilon')^{m+1}$ approximation of $f_1$. We then move to period~$2$ and continue this pattern of partition, convolutions, and truncation. In the end as we reach period $T$, $\tilde{f}_T$ would only contain feasible solutions to~\eqref{MPBKP}, and approximates $f_T$ with accumulated approximation factor $(1+\epsilon)(1+\epsilon')^{(m+1)T}\approx (1+\epsilon)(1+(m+1)T\epsilon')$. Formally, we have the following lemma which shows the approximation factor of $\tilde{f}_t$ to $f_t$.
\begin{lemma}\label{lem:fapprox2}
	Let $\tilde{f}_t$ be the functions obtained from Algorithm~\ref{alg:MPBKP'}, and let $f_t$ be defined as in~\eqref{eqn:ft}. Then, $\tilde{f}_t$ approximates $f_t$ with factor $(1+\epsilon)(1+\epsilon')^{(m+1)t}$, i.e., $\tilde{f}_t(c)\le f_t(c)\le (1+\epsilon)(1+\epsilon')^{(m+1)t}\tilde{f}_t(c)$ for all $0\le c\le c_t$.
\end{lemma}
The proof of Lemma~\ref{lem:fapprox2} relies on the following fact.
\begin{lemma}\label{lem:fapprox3}
	At any period $t$, after running the inner ``for" loop of Algorithm~\ref{alg:MPBKP'}, we have that $(1+\epsilon')^{m+1}\hat{f}_t\ge \tilde{f}_{t-1}\oplus f_{\hat{\Ical}(t)_0}\oplus f_{\hat{\Ical}(t)_1}\oplus\cdots \oplus f_{\hat{\Ical}(t)_m}$. 
\end{lemma}
\begin{proof}[Proof of Lemma~\ref{lem:fapprox3}]
	We prove by induction on $j=0,1,\ldots,m$. Base case is when $j=0$, i.e., after the first round of the inner ``for" loop, by Lemma~\ref{lem:02}, we have that $(1+\epsilon')\hat{f}_{t}\ge \tilde{f}_{t-1}\oplus  f_{\hat{\Ical}(t)_0}$. For the induction step, assume that after $j$ rounds of the inner ``for" loop, $(1+\epsilon')^j \hat{f}_{t}\ge \tilde{f}_{t-1}\oplus f_{\hat{\Ical}(t)_0}\oplus\cdots\oplus f_{\hat{\Ical}(t)_{j-1}}$, we show that after $j+1$ rounds, $(1+\epsilon')^{j+1} \hat{f}_{t}\ge \tilde{f}_{t-1}\oplus f_{\hat{\Ical}(t)_0}\oplus\cdots\oplus f_{\hat{\Ical}(t)_j}$. As a notation, we denote by $\hat{f}_t^{\rm old}$ the $\hat{f}_t$ right before the $(j+1)$th round of the inner ``for" loop, and by $\hat{f}_t^{\rm new}$ the $\hat{f}_t$ right after the $(j+1)$th round of the inner ``for" loop. Then, from Lemma~\ref{lem:02} we have that
	$
	(1+\epsilon')\hat{f}_t^{\rm new}\ge \hat{f}_t^{\rm old}\oplus f_{\hat{\Ical}(t)_j},
	$
	which implies that
	$$
	(1+\epsilon')^{j+1}\hat{f}_t^{\rm new}\ge (1+\epsilon')^j\hat{f}_t^{\rm old}\oplus f_{\hat{\Ical}(t)_j}\ge \tilde{f}_{t-1}\oplus f_{\hat{\Ical}(t)_0}\oplus\cdots\oplus f_{\hat{\Ical}(t)_{j-1}}\oplus f_{\hat{\Ical}(t)_j},
	$$
	where the second inequality follows from the induction assumption. This finishes the induction step, and thus the proof of the lemma.
\end{proof}

With Lemma~\ref{lem:fapprox3} at hand, we now prove Lemma~\ref{lem:fapprox2}.

\begin{proof}[Proof of Lemma~\ref{lem:fapprox2}]
	By the construction of $\tilde{f}_t$, it should be clear that $\tilde{f}_t\le f_t$. We prove that $(1+\epsilon)(1+\epsilon')^{(m+1)t}\tilde{f}_t\ge f_t$ by induction on $t$. Base case is when $t=1$, we have that $(1+\epsilon)(1+\epsilon')^{m+1}\tilde{f}_1 = (1+\epsilon)(1+\epsilon')^{m+1}\hat{f}_1^{c_1} \ge (1+\epsilon)\left(\tilde{f}_0\oplus f_{\hat{\Ical}(1)}\right)^{c_1} = (1+\epsilon)f_{\hat{\Ical}(1)}^{c_1}\ge  f_{\Ical(1)}^{c_1} = f_1$, where the first inequality follows from Lemma~\ref{lem:02}, and the second inequality follows from the rounding of the rewards. For the induction step, assume that $(1+\epsilon)(1+\epsilon')^{(m+1)t}\tilde{f}_t\ge f_t$, we show that $(1+\epsilon)(1+\epsilon')^{(m+1)(t+1)}\tilde{f}_{t+1}\ge f_{t+1}$.
	
	After partitioning $\Ical(t+1) = \Ical(t+1)_0\sqcup \Ical(t+1)_1\sqcup\cdots\sqcup \Ical(t+1)_m$, for any item $i\in \Ical(t+1)$, by the rounding down, we have that $(1+\epsilon)\hat{r}_i\ge r_i\ge \hat{r}_i$, which further implies that $(1+\epsilon)f_{\hat{\Ical}(t+1)_j}\ge f_{{\Ical}(t+1)_j}\ge f_{\hat{\Ical}(t+1)_j}, \forall j=0,1,\ldots,m$. Thus, 
	\begin{align*}
	    &\quad (1+\epsilon)\left(f_{\hat{\Ical}(t+1)_0}\oplus f_{\hat{\Ical}(t+1)_1}\oplus\cdots \oplus f_{\hat{\Ical}(t+1)_m}\right)\ge f_{\Ical(t+1)}
	    \ge f_{\hat{\Ical}(t+1)_0}\oplus f_{\hat{\Ical}(t+1)_1}\oplus\cdots \oplus f_{\hat{\Ical}(t+1)_m},
	\end{align*}
	which, together with the induction assumption, implies that 
	$$
	(1+\epsilon)(1+\epsilon')^{(m+1)t}\left(\tilde{f}_t\oplus f_{\hat{\Ical}(t+1)_0}\oplus f_{\hat{\Ical}(t+1)_1}\oplus\cdots \oplus f_{\hat{\Ical}(t+1)_m}\right)\ge f_t \oplus f_{\Ical(t+1)}.%\ge \tilde{f}_t\oplus f_{\hat{\Ical}(t+1)_0}\oplus f_{\hat{\Ical}(t+1)_1}\oplus\cdots \oplus f_{\hat{\Ical}(t+1)_m}
	$$
	By Lemma~\ref{lem:fapprox3}, after the inner ``for" loop in Algorithm~\ref{alg:MPBKP'}, we have that $(1+\epsilon')^{m+1}\hat{f}_{t+1}\ge \tilde{f}_t\oplus f_{\hat{\Ical}(t+1)_0}\oplus f_{\hat{\Ical}(t+1)_1}\oplus\cdots \oplus f_{\hat{\Ical}(t+1)_m}$, which implies that 
	\begin{align*}
	 (1+\epsilon)(1+\epsilon')^{(m+1)(t+1)}\hat{f}_{t+1}&\ge (1+\epsilon)(1+\epsilon')^{(m+1)(t+1)}\left(\tilde{f}_t\oplus f_{\hat{\Ical}(t+1)_0}\oplus f_{\hat{\Ical}(t+1)_1}\oplus\cdots \oplus f_{\hat{\Ical}(t+1)_m}\right)\\&\ge f_t \oplus f_{\Ical(t+1)}.
	\end{align*}
	Taking truncation on both sides, we conclude that
	\begin{align*}
	    (1+\epsilon)(1+\epsilon')^{(m+1)(t+1)}\tilde{f}_{t+1} &= (1+\epsilon)(1+\epsilon')^{(m+1)(t+1)}\hat{f}_{t+1}^{c_{t+1}}\ge \left(f_t \oplus f_{\Ical(t+1)}\right)^{c_{t+1}}= f_{t+1}.
	\end{align*}
	This finishes the induction step, and thus the proof of the lemma.
\end{proof}

Lemma~\ref{lem:fapprox2} and Proposition~\ref{prop:optimalfunc} together imply that $\tilde{f}_T(c_T)$, obtained from Algorithm~\ref{alg:MPBKP'}, approximates the optimal value of MPBKP~\eqref{MPBKP} by a factor of $(1+\epsilon)(1+\epsilon')^{(m+1)T} \approx (1+\epsilon+mT\epsilon')$. In Algorithm~\ref{alg:MPBKP'}, during each of the periods $t=1,\ldots,T$, approximately computing the $(\max,+)$-convolutions on $\hat{f}_{t}\oplus {f}_{\hat{\Ical}(t)_j}$ for all $j=0,1,\ldots,m$ take total time $\tilde{O}\left(n_t+ (m+1)/\epsilon'\right)$. Therefore, Algorithm~\ref{alg:MPBKP'} has total runtime $\tilde{O}\left(n+(m+1)T/\epsilon'\right)$. As a result, we have the following proposition. 
\begin{proposition}%[Partially restating Theorem~\ref{thm:MPBKP}]
	Taking $\epsilon = mT\epsilon'$ and $m=\tilde{\Ocal}(1/\epsilon)$, Algorithm~\ref{alg:MPBKP'} achieves $(1+\epsilon)$-approximation for MPBKP in $\tilde{O}\left(n+\frac{T^2}{{\epsilon}^3}\right)$.
\end{proposition}

\section{Alternative FPTAS for MPBKP-S}\label{simple-MPBKP-S}

%\subsection{An FPTAS for approximately optimal solution}
In this section, we provide an FPTAS for the MPBKP-S with time complexity $\mathcal{O}\left(\frac{n^2\log n}{\epsilon}\right)$. Following the classical approach for ``0-1'' knapsack problems (see, e.g., \cite{vazirani2013approximation}), we round down the reward of each item so that the optimal solution for the MPBKP under the new rounded rewards is upper bounded by some polynomial of $n$ and $1/\epsilon$, and thus the naive pseudo-polynomial dynamic program becomes a polynomial time algorithm.

We assume that the items are initially sorted and relabeled in the increasing order of their deadlines, i.e., $d_1\le d_2\le \cdots\le d_n$. Further, assume that we have a guess $P_0$ that satisfies~\eqref{P0}. Then, we choose a discretization quantum $\kappa:=\epsilon P_0/2n$ and define rounded rewards $\hat{r}_i:=\left\lfloor\frac{r_i}{\kappa}\right\rfloor_{\kappa}$. We then have $\Pcal(\mathcal{S}^*)\le \frac{4n}{\epsilon}\kappa$.

For a solution $\Scal =  \Scal(1) \cup \Scal(2) \cup \cdots \cup \Scal(T)$ where $\Scal(t)$ is the set of items with deadline $t$. Let the items in $\Scal(t)$ be indexed as $\Scal(t) = \left(i^{(t)}_1, \ldots, i^{(t)}_{S_t}\right)$ in the order in which Algorithm~\ref{alg:FPTAS_SC_1} considers them, we define the {\it rounded profit} of $\Scal$ as:
\begin{align} \label{Phat}
\hat{\Pcal}(\Scal) &= \hat{\Rcal}(\Scal) -
\sum_{t=1}^T  \sum_{k=1}^{S_t} \left \lceil B  \left( \sum_{\ell \leq k} q_{i^{(t)}_{\ell}} - \max_{0 \leq t' < t}\left\{ c_t - c_{t'} - \sum_{t'+1 \leq \tau <  t} \Qcal(\Scal(\tau)) \right\} \right)^+ \right \rceil_{\kappa} .
\end{align}
Let us also define a single period change in rounded profit for a set of items $\Scal = (i_1, \ldots, i_S)$ with knapsack capacity $c$ as:
\begin{align} \label{Phatc}
\Delta \hat{\Pcal}(\Scal,c) &= \hat{\Rcal}(\Scal) -
\sum_{k=1}^{S} \left \lceil B  \left( \sum_{\ell \leq k} q_{i_{\ell}} - c  \right)^+ \right \rceil_{\kappa} .
\end{align}

Let $\hat{A}(i,p)$ be the maximum capacity left at time $d_i$ when earning rounded profit at least $p$ using items $\{1,\ldots,i\}$ with rounded down rewards $\hat{r}$, equivalently,
\begin{align}\label{Ahat2}
\hat{A}(i,p) := \max_{\left\{\substack{\mathcal{S}\subseteq\{1,\ldots,i\}\\
		\hat{\Pcal}(\mathcal{S})\ge p}\right\}} \max_{0 \leq t' < d_i}\left\{ c_{d_i} - c_{t'} - \sum_{t'+1 \leq \tau \leq  d_i-1} \Qcal(\Scal(\tau)) \right\}.
\end{align}
If it is not possible to earn profit $p$ at time $d_i$ using items $\{1,\ldots,i\}$ with rounded down rewards, i.e., no $\mathcal{S}\subseteq\{1,\ldots,i\}$ exists such that $\hat{\Pcal}(\mathcal{S})\ge p$, then $\hat{A}(i,p)$ is labeled $-\infty$. The DP table runs for $i=1,\ldots,n$ and $p=0,\kappa, \ldots,\left\lceil\frac{4n}{\epsilon}\right\rceil\kappa$.
We then have Algorithm~\ref{alg:FPTAS_SC_1}, which returns an exact optimal solution of $\hat{\Pcal}(\Scal)$ under the rounded rewards and rounded penalties.
\begin{algorithm}[h]
	\footnotesize
	\caption{DP with rounded down rewards for MPBKP-S}
	\label{alg:FPTAS_SC_1}
	\algsetblock[Name]{Parameters}{}{0}{}
	\algsetblock[Name]{Initialize}{}{0}{}
	\algsetblock[Name]{Define}{}{0}{}
	\begin{algorithmic}[1]
		\Define \ $\kappa = \frac{\epsilon P_0 }{2n}$
		% \Define \ $\hat{r}_i = \kappa \floor{ \frac{r_i}{\kappa} }$ \Comment Round down reward
		\Define \ $\hat{r}_i = \kappa \floor{ \frac{r_i}{\kappa} }$ \Comment Round down reward
		%\Define \ $\hat{A}_{u:v} = \kappa \floor{ \frac{A_{u:v}}{\kappa} }$ \Comment Round down supply
		\Statex  \texttt{// $\hat{A}(i,p)=$ max capacity left at time $d_i$ when earning (rounded) profit at least $p$ by selecting items in $\{1,\ldots,i\}$ with rounded down rewards $\hat{r}$}
		\State Initialize $\hat{A}(0,p) = \begin{cases}
		0 & p = 0,\\
		-\infty & p > 0.
		\end{cases}$
		\For {$t=1,\ldots,T$}
		\State $i=I(t-1)+1$
		\For {$p = \left\{ 0, 1, \ldots,\left\lceil\frac{4n}{\epsilon}\right\rceil \right\} \cdot \kappa $}
		\State $\hat{A}(i,p) := \hat{A}(i-1,p)+c_{t}-c_{t-1}$ 
		\Comment If reject request $i$
		\EndFor
		\For {$\bar{p} = \left\{ 0, 1, \ldots,\left\lceil\frac{4n}{\epsilon}\right\rceil \right\} \cdot \kappa $}
		\State $p = \bar{p} + \hat{r}_i - \left\lceil B(q_i -{\color{black}\max\{0,\hat{A}(i-1,\bar{p})+(c_{t}-c_{t-1})\}})^+\right\rceil_{\kappa}$
		\State $\hat{A}(i, p) = \max\{ \hat{A}(i,p ), \hat{A}(i-1,\bar{p}) + (c_{t}-c_{t-1}) - q_i \}$ \Comment Accept $i$
		\EndFor
		\For{$p = \left\{\left\lceil\frac{4n}{\epsilon}\right\rceil,\left\lceil\frac{4n}{\epsilon}\right\rceil-1,\ldots,1  \right\}\cdot \kappa$}
		\vspace{0.1cm}
		\If {$\hat{A}(i,p-\kappa)<\hat{A}(i,p)$}
		\vspace{0.1cm}
		\State $\hat{A}(i,p-\kappa) = \hat{A}(i,p)$
		\EndIf
		\EndFor
		\For {$i=I(t-1)+2,\ldots, I(t)$}
		\For {$p = \left\{ 0, 1, \ldots,\left\lceil\frac{4n}{\epsilon}\right\rceil \right\} \cdot \kappa $}
		\State $\hat{A}(i,p) := \hat{A}(i-1, p)$ 
		\Comment If reject request $i$
		\EndFor
		\For {$\bar{p} = \left\{ 0, 1, \ldots,\left\lceil\frac{4n}{\epsilon}\right\rceil \right\} \cdot \kappa $}
		\State ${p} = \bar{p} + \hat{r}_i - \left\lceil B(q_i - {\color{black}\max\{0, \hat{A}(i-1,\bar{p})\}})^+\right\rceil_{\kappa}$
		\State $\hat{A}(i, p) = \max\{ \hat{A}(i,p ), \hat{A}(i-1,\bar{p}) - q_i \}$ \Comment Accept $i$
		\EndFor
		\For{$p = \left\{\left\lceil\frac{4n}{\epsilon}\right\rceil,\left\lceil\frac{4n}{\epsilon}\right\rceil-1,\ldots,1  \right\}\cdot \kappa$}
		\vspace{0.1cm}
		\If {$\hat{A}(i,p-\kappa)<\hat{A}(i,p)$}
		\vspace{0.1cm}
		\State $\hat{A}(i,p-\kappa) = \hat{A}(i,p)$
		\EndIf
		\EndFor
		\EndFor
		\EndFor
	\end{algorithmic}
\end{algorithm}
\begin{proof}[Proof of Correctness of Algorithm~\ref{alg:FPTAS_SC_1}]
	We show that $\hat{A}(i,p)$ returned by the algorithm satisfies \eqref{Ahat2} by induction on $i$. The base case ($i=0$) is vacuously true. Now we assume that \eqref{Ahat2} holds for all $p \in \left\{  0, 1, \ldots, \lceil 4n/\epsilon \rceil \right\}  \kappa$ and for all $k \in [i-1]$. Consider some $p \in \left\{  0, 1, \ldots, \lceil 4n/\epsilon \rceil \right\}  \kappa $, and let $\Scal^*$ be any set achieving the maximum in \eqref{Ahat2} so that $\hat{\Pcal}(\Scal) \ge p$. We will show that $\hat{A}(i,p)$ is at least the leftover capacity under solution $\Scal^*$ via case analysis:
	\begin{itemize}
		\item Case $i \notin \Scal^*$: In this case, the leftover capacity under $\Scal^*$ is the leftover capacity by $d_i$, which is the sum of leftover capacity in $\Scal^*$ by $d_{i-1}$ and $c_{d_i}-c_{d_{i-1}}$. By induction hypothesis, $\hat{A}(i-1,p)$ is no less than the leftover capacity of $\Scal^*$ by $d_{i-1}$, and therefore, by lines (7,11) and (20,24), $\hat{A}(i,p) \geq \hat{A}(i-1,p) + c_{d_i}-c_{d_{i-1}}$ which in turn is no less than the leftover capacity under $\Scal^*$ by $d_i$. By optimality of $\Scal^*$, all the inequalities must be equalities.
		\item Case $i \in \Scal^*$: Let $\Scal' = \Scal^* \setminus \{ i\}$, and let $p' = \hat{\Pcal}(\Scal')$ be its rounded profit. Then by induction hypothesis, $\hat{A}(i-1,p')$ is no less than the leftover capacity under $\Scal'$ by $d_{i-1}$. Further, by packing item $i$ in the solution corresponding to $\hat{A}(i-1,p')$, the change in profit is larger than by packing item $i$ in $\Scal'$ (the penalty is no less under $\Scal'$ since it has weakly smaller leftover capacity). Therefore, packing item $i$ in the solution corresponding to $\hat{A}(i-1,p')$ gives a solution with at least as large a rounded profit as $p$ and at least as much leftover capacity by $d_i$ as $\Scal^*$. Therefore, in turn $\hat{A}(i,p)$ is at least as much as the leftover capacity in $\Scal^*$. Since we assume $\Scal^*$ to have the largest leftover capacity with profit at least $p$, all the inequalities must be equalities.
	\end{itemize}
\end{proof}
Our next result gives the approximation guarantee for Algorithm~\ref{alg:FPTAS_SC_1}.
\begin{lemma}\label{lem:FPTAS-SC}
	Let $\mathcal{S}^*$ be the optimal solution set to the original MPBKP-S, and $P_0$ satisfy \eqref{P0}. Let $\mathcal{S}'$ denote the optimal solution set by Algorithm~\ref{alg:FPTAS_SC_1}, i.e., $\mathcal{S}'$ is the solution set corresponding to $\hat{A}(n,p^*)$ where $p^*$ is the maximum $p$ such that $\hat{A}(n,p)>-\infty$.   Then,
	$$
	\Pcal(\mathcal{S}')\geq p^* \ge (1-\epsilon)\Pcal(\mathcal{S}^*).
	$$
\end{lemma}
\begin{proof}[Proof of Lemma~\ref{lem:FPTAS-SC}]
	For any item $i$, because of rounding down, $\hat{r}_i$ is smaller than $r_i$. Also there are at most $n$ rounding ups on the penalties in $\Scal^*$, each by not more than $\kappa$. Then,
	$$
	\Pcal(\Scal^*)-\hat{\Pcal}(\Scal^*)\leq  2n\kappa.
	$$
	The dynamic programming step must return a set, $\Scal'$, at least as good as $\Scal^*$ under the new profit. Therefore,
	\begin{align*}
	\Pcal(\Scal') &\geq \hat{\Pcal}(\Scal')=p^* \geq \hat{\Pcal}(\Scal^*)\geq \Pcal(\Scal^*)-2n\kappa=\Pcal(\Scal^*)-\epsilon P_0\geq (1-\epsilon)\Pcal(\Scal^*),
	\end{align*}
	where first inequality follows because the rewards are rounded down and the penalties are rounded up in calculation of $\hat{\Pcal}$, second inequality follows because $\Scal'$ is the optimal set for objective $\hat{\Pcal}$, the third inequality follows because $|\Scal^*|\leq n$ and $T\le n$, and the last inequality follows from~\eqref{P0} that $\Pcal(\Scal^*)\geq P_0$.
\end{proof}

It remains to find $P_0$ which satisfies~\eqref{P0}. Since $\Pcal(\Scal^*)\le \bar{P}$, we can enumerate $P_0$ from $\bar{P}/2, \bar{P}/4, \bar{P}/8,\ldots$, and one of them must satisfy~\eqref{P0}. The FPTAS is presented as Algorithm~\ref{alg:FPTAS2}.

\begin{algorithm}[H]
	\footnotesize
	\caption{FPTAS for MPBKP-S in $\mathcal{O}(n^2\log n/\epsilon)$}
	\label{alg:FPTAS2}
	\algsetblock[Name]{Parameters}{}{0}{}
	\algsetblock[Name]{Initialize}{}{0}{}
	\algsetblock[Name]{Define}{}{0}{}
	\begin{algorithmic}[1]
		\State $P_0\gets {\bar{P}}$
		\State $p^*\gets 0$
		\While {$p^*<(1-\epsilon)P_0$}
		\vspace{0.1cm}
		\State $P_0\gets \frac{P_0}{2}$
		\vspace{0.1cm}
		\State	Run Algorithm~\ref{alg:FPTAS_SC_1} with the current $P_0$.
		\State $p^*\gets \max_{\left\{\substack{p\in \left\{ 0,\ldots,\left\lceil\frac{4n}{\epsilon}\right\rceil \right\} \cdot \kappa\\ \hat{A}(n,p)>-\infty}\right\}}p$
		\EndWhile
	\end{algorithmic}
\end{algorithm}

\begin{theorem}\label{thm:FPTAS2}
	Algorithm~\ref{alg:FPTAS2} is a fully polynomial approximation scheme for the MPBKP-S, which achieves $(1-\epsilon)$ factor of optimal with running time $\mathcal{O}\left(\frac{n^2\log n}{\epsilon}\right)$.
\end{theorem}
\begin{proof}[Proof of Theorem~\ref{thm:FPTAS2}]
	{\bf Time complexity: }
	When $P_0$ satisfies~\eqref{P0}, by Lemma~\ref{lem:FPTAS-SC} we have that
	$$
	p^*\ge (1-\epsilon) \Pcal(\Scal^*)\ge (1-\epsilon)P_0.
	$$
	Thus, the ``while" loop terminates when $P_0$ satisfies~\eqref{P0}, if not before $P_0$ satisfies~\eqref{P0}. When $P_0$ satisfies~\eqref{P0}, we would also have $\Pcal(\Scal^*)/2\le P_0\le \Pcal(\Scal^*)$. Therefore, the number of iterations is upper bounded by
	$$
	\text{number of iterations}\le \log\frac{\bar{P}/2}{\Pcal(\Scal^*)/2}\le \log n,
	$$
	where we have used the fact that $\bar{P}\le nP\le n\Pcal(\Scal^*)$. Since each iteration takes time $\mathcal{O}\left(n\cdot \left\lceil\frac{4n}{\epsilon}\right\rceil \right)$ we get a total time complexity of $\mathcal{O}\left(\frac{n^2\log n}{\epsilon}\right)$.
	
	{\bf Approximation ratio:} 
	When Algorithm~\ref{alg:FPTAS2} terminates, it returns the last $p^*$ and the solution set $\mathcal{S}'$ corresponding to $\hat{A}(n,p^*)$.
	If the ``while" loop terminates when $P_0>\Pcal(\Scal^*)$, i.e., it stops before $P_0$ falls below $\Pcal(\Scal^*)$, then we have that
	$$
	\Pcal(\Scal')\ge p^*\ge (1-\epsilon)P_0>(1-\epsilon)\Pcal(\Scal^*).
	$$
	Otherwise, from the time complexity analysis, we know that the ``while" loop must terminate when $P_0$ first falls below $\Pcal(\Scal^*)$, which implies that the last $P_0$ satisfies~\eqref{P0}. Then by Lemma~\ref{lem:FPTAS-SC} we again have that 
	$$
	\Pcal(\Scal')\ge (1-\epsilon)\Pcal(\Scal^*).
	$$
	In either case, the solution we obtained from Algorithm~\ref{alg:FPTAS2} achieves $(1-\epsilon)$ optimal.
	$(1-\epsilon)$ factor of $\Pcal(\Scal^*)$. 
\end{proof} 
		
	\end{appendix}

	%\end{document}  % This is where a 'short' article might terminate
	%\addtolength{\partopsep}{3mm}

	%\thispagestyle{empty}
	
	%\bibliography{sigproc}  % sigproc.bib is the name of the Bibliography in this case
	% ACM needs 'a single self-contained file'!
	
\end{document}